\documentclass[12pt,a4]{article}

\usepackage{abstract}
\usepackage{amsfonts}
\usepackage{amsmath}
\usepackage{amssymb}
\usepackage{amsthm}
\usepackage{mathabx}
\usepackage{changepage}
\usepackage{setspace}
\usepackage[multiple]{footmisc}
\usepackage[round]{natbib}
\usepackage{graphicx}
\usepackage{titlesec}
\usepackage{IEEEtrantools}
\usepackage{comment}
\usepackage{bm}
\usepackage{cancel}
\usepackage{subcaption}

\usepackage{mathtools}

\newcommand{\mbp}{\mathbf{p}}
\newcommand{\tmbp}{\ti{\mbp}}

\newcommand{\ti}{\tilde}
\newcommand{\bs}[1]{\boldsymbol{#1}}
\newcommand{\bbs}[1]{\bar{\boldsymbol{#1}}}
\newcommand{\tbs}[1]{\ti{\boldsymbol{#1}}}
\newcommand{\mbx}{\mathbf{x}}
\newcommand{\tmbx}{\ti{\mathbf{x}}}

\newcommand{\eucp}{\mathbb{R}_{+}^L}
\newcommand{\eucpp}{\mathbb{R}_{++}^L}
\newcommand{\ua}{\uparrow}
\newcommand{\da}{\downarrow}
\newcommand{\di}{\diamond}
\newcommand{\sq}{\square}
\newcommand{\E}[0]{\textbf{E}}
\newcommand{\supp}[0]{\text{Spt}}
\newcommand{\Tre}[0]{\unrhd}

\newcommand{\real}[0]{\mathbb{R}}
\newcommand{\ints}[0]{\mathbb{Z}}
\newcommand{\indep}{\perp \!\!\! \perp}
\newcommand{\ereal}[0]{[-\infty, \infty)}
\newcommand{\conv}{\text{conv}}

\newtheorem{proposition}{\normalfont\scshape Proposition}[section]
\newtheorem{theorem}{\normalfont\scshape Theorem}[section]

\newtheorem{assumption}{\normalfont\scshape Assumption}

\theoremstyle{definition}

\newtheorem{lemma}{\normalfont\scshape Lemma}
\newtheorem{definition}{\normalfont\scshape Definition}

\newtheorem{step}{\normalfont\scshape Step}

\renewcommand{\epsilon}{\varepsilon}
\renewcommand{\geq}{\geqslant}
\renewcommand{\leq}{\leqslant}

\allowdisplaybreaks

\begin{document}

\title{Revealed preference with optimal transport: \\ {\Large money pumps, bounded rationality, and preference recovery }}

\date{\normalsize{\today}}

\author{Joshua Lanier, Matthew Polisson, and John K.-H. Quah\thanks{Emails: joshua.lanier@ntu.edu.sg; matt.polisson@leicester.ac.uk; ecsqkhj@nus.edu.sg.}}

\maketitle

\begin{abstract} \linespread{1.1} \selectfont
	This paper explores the connection between two distinct notions of irrationality: the extent to which the consumer fails to maximize his utility function and the extent to which the consumer can be turned into a money pump. We show that the amount of money which can be pumped by an arbitrageur is equal to both (i) the minimum amount of money which the consumer overpays to attain his utility targets (minimum over all utility functions) and (ii) the minimum amount of wasted quasilinear utility (minimum again over all utility functions). Under the assumption that the consumer's true utility function belongs to the aforementioned argmins, we present a method for recovering this true utility function even when choices are non-optimal. Further, we show that this recovered utility function can be interpreted as belonging to a utility maximizing consumer who systematically mis-perceives prices. These results and many more are proved by exploiting a novel connection between revealed preference analysis and optimal transport.
\end{abstract}

\begin{abstract} \linespread{1.1} \selectfont
\end{abstract}

\begin{abstract} \linespread{1.1} \selectfont
\end{abstract}

\newpage

\section{Introduction}

It has often been argued that an agent's vulnerability to money pumps is a sign of bounded or limited rationality. This paper is devoted to the study of money pumps, how they can be used to measure an agent's adherence to rational principles, and how they can, under certain assumptions, be leveraged to recover the true utility function of the consumer. Many of the results presented in this paper are proved by exploiting a novel connection between the theory of revealed preferences and the theory of optimal transport.

We begin our discussion with a simple example.  Suppose that there are two goods and a consumer purchases the bundle $\mbx^1 = ( 1,2 )$ when prices are $\mbp^1 = (1,2)$ and the bundle $\mbx^2 = (2,1)$ when the prices are $\mbp^2 = (2,1)$. In period 1, the bundle $\mbx^2$ is cheaper than $\mbx^1$, while in period 2, $\mbx^1$ is cheaper than $\mbx^2$.  An arbitrageur can use these facts to turn the consumer into a money pump by selling the consumer $\mbx^1$ and purchasing $\mbx^2$ in period 1, which nets the arbitrageur $\mbp^1 \cdot ( \mbx^1 - \mbx^2 ) = 1$ and, in period 2, selling the consumer $\mbx^2$ and purchasing $\mbx^1$, netting a further $\mbp^2 \cdot ( \mbx^2 - \mbx^1 ) = 1$. At the end of the day the consumer has been pumped for 2 dollars, while there is no net change to the arbitrageur's stock of goods 1 and 2.  

More generally, suppose that there are $L$ goods and we observe a consumer who purchases the consumption bundle $\mbx^t = (x_1^t, x_2^t, \ldots, x_L^t) \in \mathbb{R}_{+}^L$ when prices are $\mbp^t = (p_1^t, p_2^t, \ldots, p_L^t) \in \mathbb{R}_{++}^L$, at some observation $t\in \{1,2,\ldots, T\}$. An arbitrageur, aware of the purchasing plans of the consumer, can sell to the consumer their desired stream of consumption bundles $\mbx^1,\mbx^2, \mbx^3, \ldots, \mbx^T$ while acquiring the bundles in a different, more economical, sequence $\mbx^{\sigma(1)}, \mbx^{\sigma(2)}, \ldots, \mbx^{\sigma(T)}$ where $\sigma$ is some permutation of $\{1,2,\ldots, T\}$. The arbitrageur has purchased and sold the same collection of bundles and so experiences no net change in inventory (i.e.\ $\sum_{t=1}^T \mbx^t = \sum_{t=1}^T \mbx^{\sigma(t)}$). Yet she earns a profit of
\begin{equation*}
	\text{MP}_{\sigma} = \sum_{t=1}^T \mbp^t \cdot ( \mbx^t - \mbx^{\sigma(t)} )
\end{equation*}
where the first term $\sum_{t=1}^T \mbp^t \cdot \mbx^{t}$ is the revenue earned by the arbitrageur from selling to the consumer and the second term $\sum_{t=1}^T \mbp^t \cdot \mbx^{\sigma(t)}$ is the cost the arbitrageur had to pay to acquire the bundles. In the event that $\text{MP}_{\sigma} > 0$ holds for some permutation $\sigma$ we say that the dataset $D=(\mbp^t,\mbx^t)_{t\leq T}$ admits a {\em money pump}. 

While the presence of a money pump means that the consumer's behavior can be exploited it also indicates a lack of cost-efficiency on the part of the consumer. Indeed, $\text{MP}_{\sigma}> 0$ indicates that the consumer could have acquired the exact same consumption bundles for less money by switching the purchasing order from $\mbx^1, \mbx^2, \ldots, \mbx^T$ to $\mbx^{\sigma(1)}, \mbx^{\sigma(2)}, \ldots, \mbx^{\sigma(T)}$. This simple swap in purchasing order results in total savings of $\text{MP}_{\sigma}$ and in this sense, the consumer's purchasing behavior is irrational or cost-inefficient.

But, how is the presence or absence of a money pump related to the other standard by which economists judge rationality, namely, utility-maximization? The absence of a money pump is also known as {\em cyclical monotonicity} and it is known that a dataset $D$ is cyclically monotone (i.e., free of money pumps) if and only if there is a {\em quasilinear rationalization} of $D$, i.e., there is a well-behaved (in the sense of being continuous, strictly increasing, and concave) utility function $U:\mathbb{R}^L_+\to\mathbb{R}$ such that $\mbx^t$ maximizes $U(\mbx)-\mbp^t\cdot \mbx$, at all $t\leq T$ (see \citet{brown-calsamiglia07}).  Alternatively, cyclical monotonicity is also equivalent to the {\em additive rationalization} of $D$, by which we mean that there is a well-behaved utility function $U$ such that 
$$\sum_{t=1}^T U(\mbx^t)\geq \sum_{t=1}^T U(\tilde \mbx^t)$$
for all $(\tilde \mbx^1,\tilde \mbx^2,\ldots,\tilde \mbx^T)\in\mathbb{R}^{LT}_{+}$ with $\sum_{t=1}^T\mbp^t\cdot \tilde\mbx^t\leq \sum_{t=1}^T\mbp^t\cdot \mbx^t$ (see \citet{browning89}). 

These characterizations of cyclical monotonicity (recall that cyclical monotonicity is equivalent to the absence of a money pump) lead naturally to the following question. For a given dataset $D=(\mbp^t,\mbx^t)_{t\leq T}$ collected from a consumer, we could work out the optimal trading strategy of the arbitrageur (the strategy $\sigma$ which nets the most money) and consequently we could work out the amount of money which would be extracted from the consumer by an arbitrageur following this optimal strategy. We use TMP (for {\em total money pump}) to denote this amount of money.  It seems sensible to regard TMP as a measure of the consumer's departure from rationality, but in what sense does the TMP actually capture irrational behavior? In particular, can TMP be understood as a measure of the degree to which behavior departs from either quasilinear or additive utility maximization?

\subsection{TMP and measures of bounded rationality}

Suppose $D$ cannot be rationalized by a well-behaved quasilinear utility function. Then, for every well-behaved $U$ there is, for some $t$, a gap between the largest possible quasilinear utility value $\max_{\mbx \in \real_{+}^L} U(\mbx) - \mbp^t \cdot \mbx$ and the utility value actually attained $U(\mbx^t) - \mbp^t \cdot \mbx^t$. \citet{allen-rehbeck26} suggest using the smallest gap (infimum over well-behaved utility functions) as a measure of the extent to which the data fails to accord with quasilinear utility maximization\footnote{Our comments here refer to their generalized approach as discussed in their Appendix B.2} 
\begin{equation} \label{eq:Q-intro}
	Q\coloneqq \inf_U  \sum_{t=1}^T \underbracket{ \max_{\mbx\in \mathbb{R}^L_+} \Big(U(\mbx)-\mbp^t\cdot \mbx \Big) }_{\text{max utility}} - \underbracket{ \Big(U(\mbx^t) - \mbp^t \cdot \mbx^t \Big) }_{\text{attained utility}}
\end{equation}
where the infimum is over the well-behaved utility functions.

In the case of the additive rationalization, we can measure the severity of departures from this property using the cost efficiency approach advocated by \citet{afriat1973}. Firstly, we observe that a dataset $D$ cannot be additively rationalized if and only if the consumer is cost inefficient in the following sense: for any well-behaved utility function $U$, the required expenditure to hit the utility target $\sum_{t=1}^TU(\mbx^t)$ is strictly lower than the amount actually spent; formally,   
$$e_U\coloneqq \min\left\{\sum_{t=1}^T \mbp^t\cdot \tilde \mbx^t: \sum_{t=1}^T U(\tilde \mbx^t)\geq \sum_{t=1}^T U(\mbx^t)\right\}$$
We can then measure the level of cost inefficiency by 
\begin{equation*} \label{eq:A-intro}
	A=\inf_U\left\{\sum_{t=1}^T \mbp^t \cdot \mbx^t - e_U\right\}
\end{equation*}
where the infimum is taken over all well-behaved utility functions. So, $A$ is the difference between the amount of money the consumer spent $\sum_{t=1}^T \mbp^t \cdot \mbx^t$ and the cheapest way of acquiring the same additive utility $e_U$.

Our first result of the paper
\begin{equation} \label{eq:main-result-intro}
	Q=A=\mbox{TMP}
\end{equation}
says that these measures are all the same.  

The result in \eqref{eq:main-result-intro} provides support for using the $\text{TMP}$ as a measure of irrationality. Not only does it measure the extent to which a consumer is vulnerable to money pumps (or equivalently the amount of money they could have saved while enjoying the same distribution of consumption) but it delivers a tight lower bound on the amount of cost inefficiency displayed by an additive utility maximizer and a tight lower bound on the amount of utility wasted by a quasilinear utility maximizer. 

It turns out that the connection between money pumps and mis-optimizing behavior extends beyond the finite data setting.

\subsection{The long run money pump}

Suppose the observations $(\mbp^t, \mbx^t)$ are random vectors following some joint distribution $\pi_0$.\footnote{Randomness of the price vector $\mbp^t$ is a natural and somewhat standard assumption in econometric analysis. The randomness of $\mbx^t$ can come from a variety of sources. First, if the consumer perfectly maximizes some quasilinear utility function his demand $\mbx^t$ is random simply because $\mbp^t$ is random. Demand could also be random due to random utility, random attention, or some other stochastic process.} The dataset $D_T = (\mbp^t, \mbx^t)_{t = 1}^T$ then is a random object and we may, for each realization, compute the total money pump $\text{TMP}_T$ (or the quasilinear waste $Q_T$ or additive cost inefficiency $A_T$ as defined above). Of course, $\text{TMP}_T = 0$ if the consumer truly maximizes a quasilinear utility function. More interestingly, if the consumer violates cyclical monontonicity with some positive probability then we may investigate the properties of the non-degenerate statistic $\text{TMP}_T$. 

Consider an arbitrageur who can neither predict the exact sequence of prices $\mbp^1, \mbp^2, \ldots, \mbp^T$ nor the exact behavior of the consumer $\mbx^1, \mbx^2, \ldots, \mbx^T$. Let us however assume that the arbitrageur has enough information to work out the \emph{marginal} distribution of demand $\nu_0$ and price $\mu_0$.\footnote{To be explicit, $P(\mbp^t \in A) = \mu_0(A)$ and $P( \mbx^t \in A ) = \nu_0(A)$ for all $t$ where $P$ is the underlying probability measure and $A$ is some Borel measurable set.} It turns out that the arbitrageur has all the information necessary to fully exploit the behavior of the consumer in the sense that perfect foreknowledge of the true data stream $D_{0} = (\mbp^t, \mbx^t)_{t \in \mathbb{N}}$ offers no additional advantage in the long run. 

Let $\Pi$ denote the collection of joint distributions with marginals $\mu_0$ and $\nu_0$. Every $\pi \in \Pi$ represents a purchasing strategy for the arbitrageur whereby she buys $\mbx$ when prices are $\mbp$ with probability $d\pi( \mbp , \mbx )$. As $\pi$ has $\nu_0$ as its marginal distribution of consumption the arbitrageur is assured to have no long-term net change in inventory provided she always sells the consumer exactly what he demands (as his consumption is drawn from $\nu_0$). Therefore, the arbitrageur can earn
\begin{equation*}
	\text{TMP}_0 = \sup_{ \pi \in \Pi } \ \E_{\pi_0}[ \mbp \cdot \mbx ] - \E_{\pi}[ \mbp \cdot \mbx ]
\end{equation*}
where $\E_{\pi_0}[\mbp \cdot \mbx]$ is the revenue the arbitrageur earns from selling to the consumer and $\E_{\pi}[ \mbp \cdot \mbx ]$ is the cost which the arbitrageur must pay to attain the consumption distribution $\nu_0$.

$\text{TMP}_T$ is the amount of money which can be extracted from the consumer by a clairvoyant arbitrageur who can perfectly forecast the data stream and $\text{TMP}_0$ is the (per period average) amount of money which can be extracted in the long run by the arbitrageur who is only aware of the marginal distributions of price and consumption. We show that, almost surely, $\lim_T \text{TMP}_T / T = \text{TMP}_0$. That is, in the long run clairvoyance offers no strategic advantage. We may think of $\text{TMP}_T / T$ as an estimator for the long run total money pump and in which case our convergence result shows strong consistency (i.e.\ almost sure convergence) of the estimator.

In the finite data context we observed in \eqref{eq:main-result-intro} a strong connection between the money pump, wasted quasilinear utility, and additive utility cost inefficiency. This relation carries over to the infinite data context. Let us define $Q_0$ to be the infimum (over all utility functions) average amount of quasilinear utility wasted by the consumer and $A_0$ to be the infimum (over all well-behaved utility function) per period cost-inefficiency displayed by an additive utility maximizer. We prove that $\text{TMP}_0 = Q_0 = A_0$. So, not only is $\text{TMP}_T$  a strongly consistent estimator for the long run money pump but also for the long run wasted quasilinear utility and additive cost inefficiency.

\subsection{Utility and price misperception}

So far, we have only discussed measures of irrationality (money pumps, wasted utility, etc). We now turn our attention to the question of recovering the utility function of the consumer. As an estimator for the utility function of the consumer we might take any well-behaved utility function $U_T$ which achieves the infimum in the definition of wasted quasilinear utility $Q$ given in \eqref{eq:Q-intro}. Such a utility function provides a better rationalization (in the sense of less wasted utility) compared to any other utility function. While this property is somewhat appealing it would be nice if we could show that $U_T$ is a consistent estimator for the true utility function given some model of mis-optimizing behavior. We are able to accomplish exactly this for a certain model of price misperception.

While the true price in period $t$ is $\mbp^t$ let us suppose that the consumer perceives the price vector to be $\tmbp^t$ where we assume that $\tmbp^t$ and $\mbp^t$ share the same distribution ($\tmbp^t \sim \mbp^t$). In the body of the paper we give several models of behavior under which this is a natural assumption. In particular, $\tmbp^t \sim \mbp^t$ if prices are stationary and the consumer considers outdated prices (so, $\tmbp^t = \mbp^{t-k}$ for some lag $k$) or if the consumer fails to update prices with some positive probability $\beta$ (so, $\tmbp^t =\mbp^t$ with probability $1-\beta$ and $\tmbp^t = \tmbp^{t-1}$ with probability $\beta$). We suppose that the consumer responds optimally given mis-perceived price vector $\tmbp$ in the sense that for each period $t$ 
\begin{equation} \label{eq:misperc-intro}
	U(\mbx^t) - \tmbp^t \cdot \mbx^t  = \max_{\mbx} \ U(\mbx) - \tmbp^t \cdot \mbx.
\end{equation}
Under the assumption $\tmbp^t \sim \mbp^t$ we are able to show that $U_T$, the utility function which achieves the infimum in the definition of $Q$ given by \eqref{eq:Q-intro}, converges pointwise to the true utility function $U$ in \eqref{eq:misperc-intro} as $T$ goes to infinity. 

To recap our asymptotic results; as our dataset $D_T$ grows we are able to recover (i) the long run money pump $\text{TMP}_0$ which happens to be equivalent to long run wasted quasilinear utility $Q_0$ and long run additive cost inefficiency $A_0$ and (ii) the true mis-optimized utility function $U$ (i.e.\ the utility function which satisfies \eqref{eq:misperc-intro}).

\subsection{Optimal Transport}
Many of our results share a strong connection with optimal transport concepts and many of our proofs rely on theorems from the optimal transport literature. The starting point for an optimal transport problem is a cost function $c: P \times X \to \real$ (where $P$ and $X$ are arbitrary Polish spaces), two probability distributions $\mu$ and $\nu$ on $P$ and $X$ respectively, and the collection $\Pi$ of all joint distributions with marginals $\mu$ and $\nu$. Following \citet{villani09} we can think of $\mu$ as a distribution over bakeries and $\nu$ as a distribution over cafes. Each bakery is capable of supplying a cafe with baked goods and the goal is to match bakeries to cafes in the most cost effective way; where the cost of matching bakery $\mbp$ to cafe $\mbx$ is $c(\mbp,\mbx)$. 

Instead of bakeries $\mbp$ being matched to cafes $\mbx$, we consider an arbitrageur who matches prices $\mbp$ to consumption bundles $\mbx$ so as to maximize profits and it is easy to show that a purchasing strategy $\pi$ maximizes the profits of our arbitrageur if and only if this $\pi$ is a solution to the optimal transport problem of matching bakeries to cafes to minimize overall costs given cost function $c(\mbp,\mbx) = \mbp \cdot \mbx$.

As mentioned, many of our results exploit this connection with optimal transport. Our result that $\text{TMP} = Q$ and its infinite data equivalent $\text{TMP}_0 = Q_0$ is proved using the Kantorovich Duality Theorem from optimal transport (see Theorem 5.10 in \citet{villani09}).\footnote{Optimal transport does not have an equivalent of our additive cost inefficiency $A$ and so showing that $Q = A$ and $Q_0 = A_0$ must done without the aid of optimal transport techniques. Actually, our proof technique is to show that $Q \geq A \geq \text{TMP}$ and $Q_0 \geq A_0 \geq \text{TMP}_0$ which is enough to prove our results given that $\text{TMP} = Q$ and $\text{TMP}_0 = Q_0$.} Our result that $\text{TMP}_T \to \text{TMP}_0$ uses continuity properties of the 2-Wasserstein metric (see Corollary 6.11 in \citet{villani09}); which is another important result from optimal transport. Our result that $U_T \to U$ uses relatively recent results from the optimal transport literature found in \citet{barrio19}.

\subsection{Related Literature}

\paragraph{Money Pumps.} Using money pumps to quantify revealed preference violations is considered in \citet{echenique2011} and \citet{smeulders13}. Their approach differs from ours in two ways. First, they only consider trading strategies of the arbitrageur which exploit violations of the generalized axiom of revealed preferences (GARP) which is a weaker requirement than that of cyclical monotonicity; violations of which, our arbitrageur seeks to exploit. Second, while we consider the \emph{total} amount of money the arbitrageur could extract by exploiting the consumer's behavior; \citet{echenique2011} consider the average (either mean or median) amount of money which could be extracted over each possible trading cycle.\footnote{By a trading cycle we mean a list of non-repeating price indices $t_1, t_2, \ldots, t_K$. The money which can be extracted from this particular cycle is $\sum_{k=1}^K \mbp^{t_k} \cdot ( \mbx^{t_k} - \mbx^{t_{k+1}} ) $ where we take $\mbx^{t_{K+1}} = \mbx^{t_1}$.} Similarly, \citet{smeulders13} considers the maximum amount of money which could be extracted over a particular trading cycle and not over the entire time period $1,2, \ldots, T$. Because our arbitrageur exploits a wider range of behavior and because our arbitrageur exploits the entire dataset instead of particular trading cycles it is clear that our $\text{TMP}$ lies above the measures considered by \citet{echenique2011} and \citet{smeulders13}.

As mentioned, \citet{allen-rehbeck26} introduce the wasted quasilinear utility measure $Q$ defined in \eqref{eq:Q-intro}. They interpret $Q$ as a measure of satisficing behavior. They show that $Q$ can be computed as the solution of a linear programming problem and show that $Q$ lies above the total amount of money which can be extracted from the consumer on any trading cycle. Specifically, their Theorem 2 in Appendix B.2 implies that 
\begin{equation*}
	Q \geq \sum_{k=1}^K \mbp^{t_k} \cdot ( \mbx^{t_k} - \mbx^{t_{k+1}} )
\end{equation*}
for any non-repeated sequence $t_1,t_2, \ldots, t_K$ where we take $t_{K+1} = t_1$.

\paragraph{Optimal Transport.}  Optimal transport, while relatively foreign to revealed preference analysis, is becoming increasingly used in economics. For a diverse range of applications and methods see the excellent books by \citet{galichon16} and \citet{chiappori17}. Also, see the recent article on optimal transport and econometrics by \citet{galichon26}. Some additional examples of of optimal transport in economics are \citet{chiappori10}, \citet{galichon22}, \citet{shum22}, \citet{galichon22},  and \citet{lin-liu24}. 

\paragraph{Revealed preference with infinite data.} Several recent papers consider revealed preference tests on infinite datasets. \citet{reny15} shows that the generalized axiom of revealed preferences (GARP) characterizes utility maximization even when the dataset has infinitely many observations. Similarly, \citet{nishimura-oK-quah17} show that (in the presence of a continuity assumption on the revealed preference relation) a property called cyclical consistency characterizes utility maximization when the utility function is required to be continuous and increasing in some given ordering $\Tre$ even when the dataset has infinitely many observations. In Section \ref{sec:pop} we consider infinite datasets where we work with a distribution (with potentially infinite support) over prices and consumption bundles. In Theorem \ref{theorem:pop-rat-quasilinear} we show that cyclical monotonicity characterizes both quasilinear utility maximization and additive utility maximization. In Section \ref{sec:pop-rat} we provide a discussion of the connection of our Theorem \ref{theorem:pop-rat-quasilinear} to the main theorem in \citet{reny15}. We also provide a result (Proposition \ref{prop:reny}) which shows that cyclical monotonicity characterizes quasilinear utility maximization when the data is an infinite collection of price-consumption bundle pairs (instead of a distribution) which is the same type of dataset considered by Reny and Nishimura et al.

\noindent {\bf Outline of the paper.}\; The remainder of the paper is organized as follows. Section \ref{sec:tmp} introduces the total money pump index $\text{TMP}$ and presents our main theorem for the finite data case (Theorem \ref{theorem:tmp}) for the $\text{TMP}$. We also present a result on welfare estimation (Theorem \ref{theorem:welfare-direct}). Section \ref{sec:pop} introduces the infinite data setup where we consider distributions over prices and consumption bundles. This section contains a result connecting the long-run money pump to the long run quasilinear utility wasted and long run additive cost inefficiency (Theorem \ref{theorem:pop-tmp-equi}) and an asymptotic result demonstrating that the sample TMP converges to the long run TMP (Theorem \ref{theorem:tmp-conv}). We also provide general results on the recovery of utility functions for the case where the distribution of consumption admits a density and a second result for the case where consumption is discrete and the distribution of price has connected support. In Section \ref{sec:price-misperception} we introduce a model of price misperception and show that any $U$ which achieves the infimum in the definition of $Q$ (i.e.\ equation \eqref{eq:Q-intro}) is a consistent estimator of the true utility function in this model. Section \ref{sec:conclusion} concludes. The appendix gives a more careful explanation of the relationship between our measure of additive cost inefficiency and the critical cost efficiency index introduced in \citet{afriat1973}, as well as all proofs omitted from the body of the paper. 

\section{The Money Pump} \label{sec:tmp}

Consider a space of alternatives $X \subseteq \mathbb{R}^L$. As in the introduction, $\bs{x} \in X$ can be interpreted as a consumption bundle but other interpretations are possible. A price vector $\mbp = (p_1, p_2, \ldots, p_L) \in \mathbb{R}_{++}^L$ is used to price the alternative $\bs{x} = (x_1,x_2, \ldots, x_L) \in X$ via $\bs{p} \cdot \bs{x}$. A \emph{dataset} is a finite collection of price vector - alternative pairs denoted $D = ( \mbp^t, \mbx^t )_{t \leq T}$ where intuitively $(\mbp^t, \mbx^t)$ means that the consumer purchased $\mbx^t$ when prices were $\mbp^t$. 

\subsection{Examples}

Some examples of different economically relevant consumption spaces $X \subseteq \mathbb{R}^L$ follows. \vspace{-10pt}

\paragraph{Consumption bundles.} As in the intro, there are $L$ goods which can be purchased in any non-negative quantity. The consumption space is then $X = \eucp$ where entry $\ell$ of the alternative $\bs{x} = (x_1,x_2, \ldots, x_L)$ is the quantity of good $\ell$ purchased by the consumer. \vspace{-10pt}

\paragraph{Discrete choice.} Suppose there are $L$ products and the consumer can purchase at most one product per observation. The consumption space is then $X = \{ \bs{0}, \bs{e}_1, \bs{e}_2, \ldots, \bs{e}_L \}$ where $\bs{e}_{\ell}$ is the vector of all 0s with a 1 in entry $\ell$. The vector $\bs{e}_{\ell}$ represents the decision to purchase product $\ell$ while the vector $\bs{0}$ represents the decision to purchase none of the $L$ products. \vspace{-10pt}

\paragraph{Stochastic choice.} Suppose again that the consumer can purchase at most one product in each period $t$ but instead of making a deterministic choice the consumer selects a distribution over the $L$ products and an outside option. The consumption space is then the simplex $X = \{ (x_1,x_2, \ldots, x_L) \in \eucp: \sum_{\ell=1}^L x_{\ell} \leq 1 \}$ where $x_{\ell}$ is the probability that the consumer buys product $\ell$ and the outside option is chosen with probability $1 - \sum_{\ell=1}^L x_{\ell}$.\vspace{-10pt}

\paragraph{Continuous / Discrete choice.} Suppose the goods can be partitioned into divisible and indivisible products. Then, $X = \mathbb{R}^{L'} \times \mathbb{N}_0^{L''}$ where $\mathbb{N}_0 = \{0,1,2,3\ldots \}$ and $L = L' + L''$.

Going forward, we shall often refer to $\bs{x} \in X$ as a consumption bundle. This is a stylistic choice. Our results hold regardless of whether the object $\bs{x}$ is a consumption bundle, a discrete choice selection, a stochastic choice vector of probabilities, or any other possible object which can be identified with a subset of $\mathbb{R}^L$ and which can be given a price via a price vector $\bs{p} \in \eucpp$.

\subsection{Utility and rationalizations.} \label{sec:utility-and-ratioanlizations}

We now discuss various models of utility maximization and what it means to rationalize (or explain) a dataset $D = ( \mbp^t, \mbx^t )_{t = 1}^T$ using these models. 
\begin{definition} \label{def:utility-defs}
	A \emph{utility function} $U$ is a map from $X$ to $\ereal$. The utility function $U$ is \emph{weakly increasing} if $\bs{x} \geq \bs{x}'$ implies $U(\bs{x}) \geq U(\bs{x}')$ and \emph{strictly increasing} if $\bs{x} > \bs{x}'$ implies $U(\bs{x}) > U(\bs{x}')$. We say $U: X \to \ereal$ is \emph{standard} if it can be extended to a concave, upper semicontinuous, and weakly increasing function on $\conv(X)$\footnote{Here $\conv(X)$ denotes the convex hull of $X$.} and $U: X \to \ereal$ is \emph{well-behaved} if it can be extended to a real-valued, concave, continuous, and strictly increasing function on $\conv(X)$. 
\end{definition}
Our definition of a utility function allows the function to take the value $-\infty$. Our main rationalization results in this section (Theorem \ref{theorem:rat-quasilinear} and Theorem \ref{theorem:tmp}) do not require this  technical allowance. However, allowing utility functions to take the special $-\infty$ value is important for Theorem \ref{theorem:welfare-direct} as well as several subsequent results in Section \ref{sec:pop}. 

For a utility function $U: X \to \ereal$ to be well-behaved it must be possible to extend $U$ to a real-valued, concave, continuous, and strictly increasing function on $\conv(X)$. Naturally, if $X$ is already convex (for instance, if $X = \eucp$) then $U$ itself must satisfy these properties and so we needn't consider extending $U$.

A dataset $D = (\mbp^t, \mbx^t)_{t \leq T}$ is \emph{rationalized} by a utility function $U: X \rightarrow \ereal$ if $U(\mbx^t) \geq (>) \ U(\mbx)$ for all $\mbx \in X$ satisfying $\mbp^t \cdot \mbx^t \geq (>) \ \mbp^t \cdot \mbx$ and $U(\bs{x}^t)$ is finite for all $t$.\footnote{What we mean by this is that 
$U(\mbx^t) \geq U(\mbx)$ for all $t$ and all $\mbx \in \mathbb{R}_+^L$ satisfying $\mbp^t \cdot \mbx^t \geq \mbp^t \cdot \mbx$, and the former inequality is strict if the latter inequality is strict.}  $D$ is \emph{additively rationalized} by $U: X \rightarrow \ereal$ if
\begin{equation*} \label{eq:add}
	\sum_{t=1}^T U(\mbx^t) \geq (>) \ \sum_{t=1}^T U(\ti{\mbx}^t), \qquad \forall (\tmbx^1, \tmbx^2, \ldots, \tmbx^T) \in X^T \text{ s.t. } \sum_{t=1}^T \mbp^t \cdot \mbx^t \geq (>) \ \sum_{t=1}^T \mbp^t \cdot \tmbx^t
\end{equation*}
and $U(\bs{x}^t)$ is finite for all $t$. An additive rationalization requires that the consumer's choices yield a higher additive (across-period) utility than any other affordable sequence of bundles $(\tmbx^1, \tmbx^2, \ldots, \tmbx^T)$. Note that if $U$ additively rationalizes the data then, for each $t$, the choice $\mbx^t$ must yield more utility than any other bundle which costs less (i.e.\ any bundle $\mbx$ satisfying $\mbp^t \cdot \mbx^t \geq \mbp^t \cdot \mbx$).  Thus, if $U$ additively rationalizes the data then $U$ rationalizes the data as well. Clearly, the converse need not hold. The dataset $D$ is \emph{quasilinear rationalized} by a utility function $U: X \rightarrow \ereal$ if, for all $t$, 
\begin{equation*}
	U(\mbx^t) - \mbp^t \cdot \mbx^t \geq \ U(\mbx) - \mbp^t \cdot \mbx,\qquad \qquad \forall \mbx \in X
\end{equation*}
and $U(\bs{x}^t)$ is finite for all $t$. A quasilinear utility maximizer seeks to maximize utility net of expenditure. Note that in the case of a quasilinear rationalization, the consumer is not constrained by a budget set, but expenditure leads to dis-utilty and has the effect of restraining purchases.  

\subsection{Money Pump and Cyclical Monotonicity} \label{sec:mp-cm}

Dataset $D = (\mbp^t, \mbx^t)_{t \leq T}$ contains a {\em money pump} if there is a permutation of $\{1,2,\ldots, T\}$, denoted $\sigma$, so that 
\begin{equation*}
	\text{MP}_{\sigma} = \sum_{t=1}^T \mbp^{t} \cdot ( \mbx^{ t} - \mbx^{ \sigma(t) } )
\end{equation*}
is strictly positive. A dataset $D$ that is free of money pumps is said to satisfy {\em cyclical monotonicity}. The number $\text{MP}_{\sigma}$ is the amount of money an arbitrageur earns by, in each period $t$, selling the bundle $\mbx^t$ to the consumer and buying the bundle $\mbx^{\sigma( t )}$. This strategy earns the consumer a net of $\mbp^t \cdot ( \mbx^t - \mbx^{\sigma(t)} )$ in period $t$ and, because $\sigma$ is a permutation, there can be no net change in the inventory of the arbitrageur in the sense that  $\sum_{t=1}^T \mbx^t = \sum_{t=1}^T \mbx^{\sigma(t)}$.

We define the \emph{total money pump} (TMP) as the amount of money which would be extracted from the consumer by an arbitrageur following an optimal trading strategy. That is, 
\begin{equation*} 
	\text{TMP} = \sup_{\sigma} \sum_{t=1}^T \mbp^{t} \cdot ( \mbx^{ t} - \mbx^{ \sigma(t) } )
\end{equation*}
where the supremum is taken over all permutations $\sigma$. 

It is not hard to show that if $D$ is quasilinear rationalized by $U$ then $D$ is also additively rationalized by $U$; on the other hand, there is no guarantee that if $D$ is additively rationalized by $U$, then the same $U$ provides a quasilinear rationalization of $D$. The following result shows, among other things, that in fact these two types of rationalizations are empirically equivalent, and each are equivalent to $\text{TMP} = 0$.

\begin{theorem} \label{theorem:rat-quasilinear}
	For any dataset $D = (\mbp^t, \mbx^t)_{t \leq T}$ the following are equivalent.
	\begin{enumerate}
		\item $\text{{\em TMP}} = 0$.
		\item The dataset $D$ satisfies cyclical monotonicity.
		\item The dataset $D$ can be additively rationalized by a utility function $U$.
		\item The dataset $D$ can be additively rationalized by a well-behaved utility function $U$.
		\item The dataset $D$ can be quasilinear rationalized by a utility function $U$.
		\item The dataset $D$ can be quasilinear rationalized by a well-behaved utility function $U$. 
	\end{enumerate}
\end{theorem}

The claims in Theorem \ref{theorem:rat-quasilinear} are not new in the sense that they can, without too much difficulty, be pieced together using existing results.\footnote{The existing results to which we refer consider the special case where $X = \eucp$. That said, there is no difficulty in using the existing proof techniques to extend the existing results to allow for arbitrary non-empty consumption spaces $X \subseteq \mathbb{R}^L$.} The equivalence between statements 1 and 2 is true basically by definition. The equivalence between statements 2, 3, and 4 is found in \citet{browning89} and the equivalence between statements 2, 5, and 6 is found in \citet{brown-calsamiglia07}. To keep this article reasonably self-contained, we provide a proof of Theorem  \ref{theorem:rat-quasilinear} in the Appendix.

The equivalence between statements 3 and 4 and the equivalence between items 5 and 6 shows that there are no additional restrictions placed on the data by assuming that $U$ is well-behaved in either the additive rationalization or the quasilinear rationalization. 

The equivalence between 1, 4, and 6 shows that there is a tight relationship between the existence of a money pump and additive and quasilinear rationalizations. This suggests that the value of the total money pump may be useful as a measure of the degree to which the consumer fails to act as an additive or quasilinear utility maximizer. We address this issue in the next subsection.  

\subsection{$\text{TMP}$ as a measure of rationality} \label{sec:tmp-rat}

By definition, a dataset $D = (\mbp^t, \mbx^t)_{t \leq T}$ can be additively rationalized if the observations, when taken as a whole, maximize an overall utility function  $V: X^{T} \to \ereal$ that is additive across bundles at each observations, i.e.,  
$V(\tilde\mbx^1,\tilde\mbx^2,\ldots,\tilde\mbx^T)=\sum_{t=1}^T U(\tmbx^t)$, subject to the overall expenditure not exceeding the consumer's total expenditure, which is $\sum_{t=1}^T\mbp^t\cdot\mbx^t$. When $D$ cannot be additively rationalized, how should we measure the extent of the violation?  The cost-efficiency approach proposed by \citet{afriat1973} measures the amount of money which could have been saved by the consumer were they to have acted perfectly in line with the utility maximization hypothesis under investigation (the additive utility model in our case).\footnote{Appendix \ref{sec:cost-efficiency} discusses in greater detail the relationship between $A$ and Afriat's critical cost efficiency index.}  For a utility function $U$ let $e_U$ denote the smallest amount of money for which a consumer with additive (across periods) utility function $U$ could obtain utility $\sum_{t=1}^T U(\mbx^t)$. That is,
\begin{equation*}
	e_U = \inf \left\{ \sum_{t=1}^T \mbp^t \cdot \tmbx^t: \sum_{t=1}^T U(\tmbx^t) \geq \sum_{t=1}^T U(\mbx^t) \right\}
\end{equation*}
The \emph{additive cost inefficiency} displayed by the consumer is the difference between the amount that the consumer actually spent and the smallest amount of money they could have spent to achieve the same utility. More formally, the additive cost inefficiency is the number
\begin{equation*} \label{eq:additive}
	A = \inf_U \left( \sum_{t=1}^T \mbp^t \cdot \mbx^t - e_U \right)
\end{equation*}
where the infimum is taken over all utility functions for which $U(\bs{x}^t)$ is finite for all $t$.

Next, let us suppose that we would like to measure the extent to which the consumer has failed to act as a quasilinear utility maximizer. One approach, which might be termed the utility efficiency approach, is to measure the additional utility which the consumer could have derived had they acted perfectly in line with the model of utility maximization under investigation. Of course, this approach only makes sense when utility is cardinal (in particular, the utility function must be identified up to translation by the consumer's behavior). If the utility function is only identified up to monotonic transformation then it makes no sense to talk about differences in utility. As the quasilinear utility function is indeed cardinal in the requisite sense, it is reasonable to define the \emph{quasilinear utility inefficiency} displayed by the consumer as
\begin{equation*} 
	Q = \inf_U \left( \sum_{t=1}^T \sup_{\mbx \in X} \left[ U(\mbx) - \mbp^t \cdot \mbx \right] - \left[ U(\mbx^t) - \mbp^t \cdot \mbx^t \right] \right)
\end{equation*}
where the infimum is over all utility functions $U$ for which $U(\bs{x}^t)$ is finite for all $t$. The term in the supremum represents the largest amount of quasilinear utility which could have been attained by the consumer in period $t$ whereas the rightmost term represents the amount of quasilinear utility actually achieved. The object $Q$ was introduced in \citet{allen-rehbeck26} as a measure of deviation from quasilinear utility maximization.\footnote{What we call $Q$ corresponds to the ``minimum deviations'' considered in Appendix B.2 of \citet{allen-rehbeck26} when using (in their language) the aggregator $f(e_1, e_2, \ldots, e_T) = \sum_{t=1}^T e_t$.} 

Importantly, \citet{allen-rehbeck26} show that $Q$ is easy to calculate. In particular, $Q$ is equal to $\bar\varepsilon$, the solution value for the linear programming problem of finding $(u_1,u_2,\ldots, u_T) \in \mathbb{R}^T$ and $(\varepsilon_1, \varepsilon_2, \ldots, \varepsilon_T) \in \mathbb{R}_+^T$ to solve
\begin{IEEEeqnarray}{rCl}
	\min & \quad & \sum_{t=1}^T \varepsilon_t \nonumber \\
	\text{s.t.} & \quad & u_s \leq u_t + \mbp^t \cdot ( \mbx^s - \mbx^t ) + \varepsilon_t, \qquad \text{ for all } s,t \label{eq:lp}
\end{IEEEeqnarray}

To recap, we have introduced three distinct ways of quantifying deviations from cyclical monotonicity. It turns out that they are all the same.

\begin{theorem} \label{theorem:tmp}
For any dataset $D = (\mbp^t, \mbx^t)_{t \leq T}$, 
	\begin{equation*} \label{eq:tmp-result}
		\text{{\em TMP}} = A = Q = \bar{\varepsilon}.
	\end{equation*}
	Moreover, there exists a well-behaved utility function $U: X \rightarrow \mathbb{R}$ that attains the infimum in the definitions of $A$ and $Q$. 
\end{theorem}
As noted above, the result $Q = \bar{\varepsilon}$ is shown in \citet{allen-rehbeck26}. We include it here and provide a proof for the sake of completeness. 

To give some insight into the proof of Theorem \ref{theorem:tmp} let us focus on how we show that $\text{TMP} = Q$ (the idea behind showing that $\text{TMP} = A$ is similar). The ``easy direction'' is showing that $Q \geq \text{TMP}$. Indeed, for any $U$ and permutation $\sigma$ we have
\begin{IEEEeqnarray*}{rCl}
	\sum_{t=1}^T \sup_{\mbx \in X} \left[ U(\mbx) - \mbp^t \cdot \mbx \right] - \left[ U(\mbx^t) - \mbp^t \cdot \mbx^t \right] & \geq & \sum_{t=1}^T \left[ U( \mbx^{\sigma(t)} ) - \mbp^t \cdot \mbx^{\sigma(t)} \right] - \left[ U(\mbx^t) - \mbp^t \cdot \mbx^t \right] \\
	& = & \sum_{t=1}^T \mbp^t \cdot ( \mbx^t - \mbx^{\sigma(t)} )
\end{IEEEeqnarray*}
and so the amount of quasilinear utility wasted for any $U$ is always greater than the amount of money which can be pumped for any $\sigma$. Thus, $Q \geq \text{TMP}$.

Showing that $\text{TMP} \geq Q$ is more delicate. The key insight we utilize is that for any permutation $\sigma$ which achieves the supremum in the definition of $\text{TMP}$ it happens that the permuted dataset $D_{\sigma} = ( \mbp^t, \mbx^{\sigma(t)} )_{t \leq T}$ satisfies cyclical monotonicity. In other words, the purchasing behavior of the arbitrageur, which is given by $D_{\sigma}$, must satisfy cyclical monotonicity.  Once this fact is established Theorem \ref{theorem:rat-quasilinear} can be applied to show that there exists a well-behaved utility function $U$ which rationalizes the permuted data. It can then be shown that the amount of money pumped via $\sigma$ is weakly greater than the amount of quasilinear utility wasted according to $U$ which establishes that $\text{TMP} \geq Q$. 

Any well-behaved $U$ which achieves the infimum in the definition of $Q$ also achieves the infimum in the definition of $A$. The converse is not true. In particular if $U$ achieves the infimum in $Q$ then $ k U$ achieves the infimum in $A$ for \emph{any} $k > 0$ but $k U$ achieves the infimum in $Q$ only for $k = 1$. In Section \ref{sec-app:Q-A} of the Online Appendix we investigate further the connection between the utility functions which achieve the infimum in the definition of $A$ and those which achieve the infimum in the definition of $Q$.

Recall that from items 3-6 of Theorem \ref{theorem:rat-quasilinear} we learned that when dealing with additive or quasilinear rationalizations there are no additional restrictions imposed on the data by requiring that the rationalizing utility function is well-behaved. This insight is preserved in Theorem \ref{theorem:tmp} in the sense that the infima in the definitions of $A$ and $Q$ can always be achieved by well-behaved utility functions and thus we could have defined $A$ and $Q$ using infima over the collection of well-behaved utility functions without changing the content of these objects. 
 
\subsection{Welfare} \label{sec:welfare}

Suppose we would like to estimate the difference in utility between two bundles $\bs{x}^{\diamond}, \bs{x}^{\square} \in X$. While we do not know the consumer's utility function we are able to identify some collection of utility functions $\mathcal{U}$ to which the consumer's true utility function belongs. This collection $\mathcal{U}$ allows us to calculate the largest and smallest possible utility differences\footnote{The left equation in \eqref{eq:utility-diffs} uses the convention $\infty - \infty = -\infty$ and the right equation uses $\infty - \infty = \infty$. These conventions allow us to ignore utility functions which satisfy $U( \bs{x}^{\diamond} ) = U(\bs{x}^{\square}) = - \infty$.}$^,$\footnote{Similar ``utility differences'' objects are considered in \citet{allen-rehbeck20} for a different model of approximate quasilinear utility maximization.}
\begin{equation} \label{eq:utility-diffs}
	U_{\ua}(\mbx^{\di}, \mbx^{\sq}; \mathcal{U}) = \sup_{U \in \mathcal{U}} \Big\{ U( \mbx^{\di} ) - U(\mbx^{\sq}) \Big\} \quad \text{ and } \quad U_{\da}(\mbx^{\di}, \mbx^{\sq}; \mathcal{U}) = \inf_{U \in \mathcal{U} } \Big\{ U(\mbx^{\di}) - U(\mbx^{\sq}) \Big\}
\end{equation}
The object on the left is the maximum possible utility difference between $\bs{x}^{\di}$ and $\bs{x}^{\sq}$ under the assumption that the true utility function belongs to $\mathcal{U}$. The object on the right is the smallest utility difference. The concepts are connected via the formula $U_{\ua}(\mbx^{\di}, \mbx^{\sq}; \mathcal{U}) = -U_{\da}( \mbx^{\sq}, \mbx^{\di}; \mathcal{U} )$ and thus it suffices to focus attention on the maximum utility difference as any result on $U_{\ua}(\mbx^{\di}, \mbx^{\sq}; \mathcal{U})$ can easily be transformed to a result on $U_{\da}(\mbx^{\di}, \mbx^{\sq}; \mathcal{U})$.

When the behavior of the consumer $D = ( \bs{p}^t, \bs{x}^t )_{t \leq T}$ satisfies cyclical monotonicity then it makes sense to let $\mathcal{U}$ in \eqref{eq:utility-diffs} be the collection of utility functions which quasilinear rationalize $D$. When no such rationalizing $U$ exists we can generalize the approach by taking $\mathcal{U}$ to be the utility functions which achieve the infimum in the definition of $Q$.\footnote{See Section \ref{sec:argmin-char} in the appendix for an alternative way of characterizing this collection of utility functions.}

It turns out that we can evaluate the utility differences in \eqref{eq:utility-diffs} by solving a linear programming problem. So, let $\mathcal{T} = \{ 1,2,\ldots, T, \diamond, \square \}$ and let $\hat{\mathcal{U}}_{\ua}( \bs{x}^{\diamond}, \bs{x}^{\square} )$ denote the solution value for the linear programming problem of finding numbers $\{ u^t \}_{ t \in \mathcal{T} }$, non-negative numbers $\{ \varepsilon^t \}_{t \in \mathcal{T}}$, and vectors $\mbp^{\diamond}, \mbp^{\square} \in \eucp$ to solve
\begin{IEEEeqnarray}{rCrCl}
	\sup & \quad & u^{\diamond} & - & u^{\square} \nonumber \\
	\text{s.t.} & & u^s & \leq & u^t + \mbp^t \cdot ( \mbx^s - \mbx^t ) + \varepsilon^t, \qquad \forall s, t \in \mathcal{T}  \label{eq:lp-ineqs-direct} \\
	&& \sum_{t=1}^T \varepsilon^t & = & \text{TMP} \nonumber \\
	&& \varepsilon^{\diamond} & = & \varepsilon^{\square} = 0 \nonumber
\end{IEEEeqnarray}
where $\text{TMP}$ denotes the total money pump for the dataset $D$. Note that the inequalities in \eqref{eq:lp-ineqs-direct} are exactly the same as the ones appearing in \eqref{eq:lp} except that $\mathcal{T}$ includes not only the observed indices $1,2,\ldots, T$ but also the new observations: $\di, \sq$. This linear program turns out to deliver our object of interest: $U_{\ua}(\bs{x}^{\di}, \bs{x}^{\sq}, \mathcal{U})$.
\begin{theorem} \label{theorem:welfare-direct}
	Let $D = (\mbp^t, \mbx^t)_{t \leq T}$ be a dataset, let $\mbx^{\diamond}, \mbx^{\square} \in X$ be two consumption bundles, let $\mathcal{U}_{wb}$ and $\mathcal{U}_{st}$ denote the well-behaved and standard utility functions (see Definition \ref{def:utility-defs}), respectively which achieve the infimum in the definition of $Q$. Then,
	\begin{equation} \label{eq:u-diffs}
		U_{\ua}( \mbx^{\diamond}, \mbx^{\square}; \mathcal{U}_{wb} ) = U_{\ua}( \mbx^{\diamond}, \mbx^{\square}; \mathcal{U}_{st} ) = \hat{U}_{\ua}( \mbx^{\diamond}, \mbx^{\square} )
	\end{equation}
	and the supremum in the definition of $U_{\ua}( \mbx^{\diamond}, \mbx^{\square}; \mathcal{U}_{st} ) $ is attained. 
\end{theorem}
The theorem guarantees that there is some standard utility function $U: X \to \ereal$ which achieves the supremum in the definition of $U_{\ua}( \mbx^{\diamond}, \mbx^{\square}; \mathcal{U}_{st} )$. This statement cannot be strengthened to require that the utility function is well-behaved. For a simple example, suppose there is one good and we observe the consumer purchase $1$ unit when prices are $1$. We would like to know the maximum utility difference between $1$ unit of consumption and $0$ units of consumption. For each $n$ let $U_n(x) = \min( x-1, n(x-1) )$. Each $U_n$ is a well-behaved utility function which quasilinear rationalizes the consumer's behavior (and thus it attains the infimum in the definition of $Q$). Yet, $U_n( 1 ) - U_n(0) = n$ which can be sent to infinity by sending $n \to \infty$. Thus, $U_{\ua}( 1, 0; \mathcal{U}_{wb} ) = \infty$ and no well-behaved utility function will achieve this value. 

Theorem \ref{theorem:welfare-direct} is no longer true when considering non concave utility functions. That is, in general $U_{\ua}( \mbx^{\diamond}, \mbx^{\square}; \mathcal{U}_{wb} ) \neq U_{\ua}( \bs{x}^{\diamond}, \bs{x}^{\square}; \mathcal{U} )$ when $\mathcal{U}$ is the collection of continuous and strictly increasing (but not necessarily concave) utility functions which achieve the infimum in the definition of $Q$. For an example, suppose there is a single good and two observations $D = ( (p^1, x^1), (p^2, x^2) )$ where $p^1 = p^2 = 1$, $x^1 = 1$ and $x^2 = 3$. Take $x^{\diamond} = x^2 = 3$ and $x^{\square} = 2$. Note that $D$ can be quasilinear rationalized (for instance, by $U(x) = x$) and therefore the supremum in the definition of $U_{\ua}( \mbx^{\diamond}, \mbx^{\square}; \mathcal{U} )$ involves only utility functions which quasilinear rationalize $D$. If $U$ quasilinear rationalizes $D$ then $U(x^1) - x^1 \geq U( x^2 ) - x^2$ and $U( x^2 ) - x^2 \geq U(x^1) - x^1$ and therefore $U(1) - 1 = U(3) - 3$. Rearranging we have $U( 3 ) - U(1) = 2$. If $U$ is concave then $U( 3 ) - U(2) \leq U(2)- U(1)$ and so $U(3) - U(2) \leq 1$. This inequality is achieved as an equality for $U(x) = x$ and thus the supremum utility difference is $1$ when requiring concavity. On the other hand, let $U_n$ be a sequence of continuous and increasing functions which, for each $n$, satisfies (i) $U_n(x) = x$ for $x \notin ( 1,3 )$, (ii) $x \geq U_n(x)$ for each $x \in (1,3)$, and (iii) $U_n( 2 ) = 1 + 1/n$. Now, we have $U_n(3) - U_n(2) \to 2$ from which it is clear that the supremum utility difference is $2$ without the requirement of concavity.

Theorem \ref{theorem:welfare-direct} is similar in spirit to the results in the ``Ordinal comparisons of consumption bundles'' section of \citet{varian1982}. In this section Varian considers a dataset $D = (\mbp^t, \mbx^t)_{t \leq T}$ which can be rationalized by a well-behaved utility function (equivalently $D$ satisfies the generalized axiom of revealed preferences). Given two bundles $\bs{x}^{\diamond}$ and $\bs{x}^{\square}$ he presents a linear programming problem which determines whether or not $U(\bs{x}^{\diamond}) > U( \bs{x}^{\square} )$ for all well-behaved utility functions $U$ which rationalize the data. 

There are two key distinctions between Varian's results and our own. First, Varian's starting point is a dataset $D$ which can be rationalized by a well-behaved utility function whereas our Theorem \ref{theorem:welfare-direct} places no requirements on the dataset. Second, utility in Varian's case has no natural units and so it is meaningless to discuss the difference in utility obtained between $\bs{x}^{\diamond}$ and $\bs{x}^{\square}$. In contrast, Theorem \ref{theorem:welfare-direct} concerns quasilinear utility wherein expenditure supplies a natural unit of utility. For this reason we are able to meaningfully put bounds on the difference in utility obtained between two bundles.

\section{The Long Run Money Pump} \label{sec:pop}

\subsection{The data}

In this section we move beyond datasets with finitely many observations and instead assume that we observe a joint distribution $\pi_0$ over prices and the consumption bundles purchased by the consumer $(\bs{p}, \bs{x})$. In other words, in this section a dataset $\pi_0$ is a distribution on $\eucpp \times X$ which represents the observed distribution of prices and consumption. Naturally, such data $\pi_0$ need not be generated by some finite sequence of observations $D = (\bs{p}^t, \bs{x}^t)_{t \leq T}$ and so the type of data we consider in this section is more general than what we considered in Section \ref{sec:tmp}. 

Suppose the distribution $\pi_0$ on $\eucpp \times X$ has marginals $\mu_0$ and $\nu_0$.\footnote{That is, $\mu_0$ is the distribution of prices and $\nu_0$ is the distribution of consumption.} As before, we imagine that the consumer has no control over the prices that he faces and so his decision is merely to purchase consumption bundles based on the realized prices. In other words, the consumer's job is to select a joint distribution $\pi$ over $\eucpp \times X$ whose first marginal is the distribution of prices $\mu_0$. Let $\Pi(\mu_0)$ denote the collection of all such joint distributions (so, $\pi \in \Pi(\mu_0)$ if and only if $\mu_0(A) = \pi( A \times X )$ for all Borel $A$).

\subsection{Utility and rationalizations}

The data $\pi_0$ with marginals $\mu_0$ and $\nu_0$ is \emph{additively rationalized} by a measurable utility function $U: X \to \ereal$ if\footnote{We define $\bs{E}_{\pi}[ U(\bs{x}) ] = -\infty$ whenever $\bs{E}_{\pi}[ \min( U(\bs{x}), 0) ] = -\infty$. So, $\bs{E}_{\pi}[ U(\bs{x}) ]$ is defined for all $\pi \in \Pi(\nu_0)$. All our results will assume that $\mu_0$ and $\nu_0$ have finite second moments and thus $\E_{\pi_0}[ \mbp\cdot\mbx ]$ is always finite.}$^,$\footnote{The notation $\E_{\pi}[ U(\mbx) ]$ is shorthand for $\int U(\mbx) \ d \pi(\mbp,\mbx)$ and similarly the notation $\E_{\pi}[ \mbp \cdot \mbx ]$ is shorthand for $\int \mbp \cdot \mbx \ d \pi(\mbp, \mbx)$.}
\begin{equation*}
	\E_{\pi_0}[ U( \mbx ) ] \geq (>) \ \E_{\pi}[ U(\mbx) ] , \qquad \forall \pi \in \Pi(\mu_0) \text{ s.t. } \E_{ \pi_0 }[ \mbp \cdot \mbx ] \geq (>) \ \E_{\pi}[ \mbp \cdot \mbx ]
\end{equation*}
and $\E_{\pi_0}[ U(\bs{x}) ]$ is finite. An additive rationalization requires that the consumer's choices $\pi_0$ yield a higher additive utility than any other affordable distribution of purchases $\pi \in \Pi(\mu_0)$. The data $\pi_0$ is \emph{quasilinear rationalized} by a measurable utility function $U$ if\footnote{Here $\supp(\pi_0)$ denotes the support of $\pi_0$. That is, $\supp(\pi_0)$ is the smallest closed set $A$ satisfying $\pi_0(A) = 1$.}
\begin{equation*}
	U( \mbx ) - \mbp \cdot \mbx = \max_{ \tmbx \in X } \ \left\{ U( \tmbx ) - \mbp \cdot \tmbx \right\}, \qquad  \text{ for all } ( \bs{p}, \bs{x} ) \in \supp(\pi_0)
\end{equation*}
and $\E_{\pi_0}[ U(\bs{x}) ]$ is finite. The quasilinear rationalization requires that, for each $(\bs{p}, \bs{x})$ in the support of $\pi_0$, the choice $\bs{x}$ given prices $\bs{p}$ yields more utility net of expenditure than any other choice $\tbs{x} \in X$.

\subsection{Money Pump and Cyclical Monotonicity} \label{sec:pop-rat}

Suppose an arbitrageur is aware of the purchasing behavior of the consumer $\pi_0$ and considers an arbitrage strategy to maximize profits. The arbitrageur sells to the consumer (and so earns a revenue of $\int \bs{p} \cdot \bs{x} \ d \pi_0( \bs{p}, \bs{x} ) $) and makes purchases according to a joint distribution $\pi$ over price-consumption pairs $(\mbp, \bs{x})$. Naturally, we assume that the arbitrageur faces the same distribution of prices as the consumer and thus her strategy $\pi$ must belong to $\Pi(\mu_0)$. Moreover, to ensure that the arbitrageur has no net change in her long run stock of goods she chooses a $\pi$ whose marginal distribution of consumption is equal to that of the consumer's (i.e.\ $\pi( \eucpp \times A ) = \nu_0(A)$ for all Borel $A \subseteq X$). In other words, the arbitrageur's purchasing behavior $\pi$ is a \emph{coupling} of $\mu_0$ and $\nu_0$ where a coupling is a joint distribution with marginals $\mu_0$ and $\nu_0$, respectively. Let $\Pi(\mu_0,\nu_0)$ denote all such couplings of $\mu_0$ and $\nu_0$.

An arbitrageur who employs strategy $\pi \in \Pi( \mu_0, \nu_0 )$ nets
\begin{equation*}
	\text{MP}_0(\pi) = \textbf{E}_{\pi_0}[ \mbp \cdot \mbx ] - \E_{\pi}[ \mbp \cdot \mbx ]
\end{equation*}
and the optimal strategy yields\footnote{For general $\pi_0$ it may be the case that the definition of $\text{TMP}_0$ involves the forbidden operation $\infty - \infty$. In all of our results we assume that $\pi_0$ has finite variance which rules out this possibility.}
\begin{equation} \label{eq:pop-tmp}
	\text{TMP}_0 = \sup_{ \pi \in \Pi(\mu_0, \nu_0) } \Big\{ \E_{\pi_0}[ \mbp \cdot \mbx ] - \E_{\pi}[ \mbp \cdot \mbx  ] \Big\}.
\end{equation}
The data $\pi_0$ on $\eucpp \times X$ is said to satisfy \emph{cyclical monotonicity} if each finite dataset $D = ( \mbp^t, \mbx^t )_{t \leq T} \subseteq \supp( \pi_0 )$ satisfies cyclical monotonicity. The following shows that, not only does cyclical monotonicity preclude money pumps, but it also characterizes additive and quasilinear rationalizability.

\begin{theorem} \label{theorem:pop-rat-quasilinear}
	Let $\pi_0$ be a dataset with finite variance\footnote{That is, $\int || (\bs{p}, \bs{x}) ||^2 \ d\pi_0(\bs{p}, \bs{x}) < \infty$} and with marginals $\mu_0$ and $\nu_0$. The following are equivalent.
	\begin{enumerate}
		\item $\text{TMP}_0 = 0$. 
		\item The data $\pi_0$ satisfies cyclical monotonicity.
		\item The data $\pi_0$ can be additively rationalized by some $U$.
		\item The data $\pi_0$ can be additively rationalized by some standard $U$.
		\item The data $\pi_0$ can be quasilinear rationalized by some $U$.
		\item The data $\pi_0$ can be quasilinear rationalized by some standard $U$.
	\end{enumerate}
\end{theorem}
Theorem \ref{theorem:pop-rat-quasilinear} is an obvious corollary of our next theorem (Theorem \ref{theorem:pop-tmp-equi}) and so we do not provide a proof of Theorem \ref{theorem:pop-rat-quasilinear}.

The main takeaways of Theorem \ref{theorem:pop-rat-quasilinear} are (i) cyclical monotonicity characterizes both quasilinear and additive rationalization and (ii) there are no additional testable implications of requiring the rationalizing utility function $U$ to be standard. This can be contrasted with Theorem \ref{theorem:rat-quasilinear} which reached the stronger conclusion that there are no additional testable implications of requiring the rationalizing utility function to be well-behaved. In particular, it is not hard to show that requiring $U$ in Theorem \ref{theorem:pop-rat-quasilinear} to be strictly increasing or continuous places additional requirements on the data beyond cyclical monotonicity.

In the finite dataset world of Section \ref{sec:tmp} a dataset was defined to satisfy cyclical monotonicity when $\text{TMP} = 0$. This is not the case in the infinite data world. Here, a dataset has $\text{TMP}_0 = 0$ when there is no coupling $\pi \in \Pi(\mu_0, \nu_0)$ which can extract money from the consumer (i.e.\ $\text{MP}_0(\pi) \leq 0$ for all $\pi \in \Pi(\mu_0, \nu_0)$). On the other hand, the dataset satisfies cyclical monotonicity when there is no finite sequence $(\bs{p}^1, \bs{x}^1), (\bs{p}^2, \bs{x}^2), \ldots, ( \bs{p}^T, \bs{x}^T )$ in $\supp(\pi_0)$ so that: $\sum_{t=1}^T \bs{p}^t \cdot ( \bs{x}^{t} - \bs{x}^{t-1} ) > 0$ where $\bs{x}^0 = \bs{x}^{T}$. While it seems intuitively clear that $\text{TMP}_0 = 0$ implies cyclical monotonicity, the converse is by no means obvious. Thus, the equivalence of items 1 and 2 in Theorem \ref{theorem:pop-rat-quasilinear} is non-trivial. 

Our Theorem \ref{theorem:pop-rat-quasilinear} is similar in spirit to Theorem 1 in \citet{reny15}. While our theorem characterizes additive and quasilinear utility maximization given a distribution over $(\bs{p}, \bs{x})$ Reny characterizes utility maximization subject to a budget constraint given a (possibly infinite) set of price-consumption pairs $(\bs{p}, \bs{x})$. The starting point of Reny's result is a non-empty collection $D \subseteq \eucp \times \eucp$ where each $( \bs{p}, \bs{x} ) \in D$ is interpreted as the observation that a consumer bought the bundle $\bs{x}$ when prices where $\bs{p}$. A utility function $U$ rationalizes $D$ if, for each $(\bs{p}, \bs{x}) \in D$ the utility obtained from the bundle purchased exceeds the utility of any other cheaper bundle. That is, $U$ rationalizes $D$ provided, for all $(\bs{p}, \bs{x}) \in D$ we have $U( \bs{x} ) \geq (>) \ U( \bbs{x} )$ whenever $\bs{p} \cdot \bs{x} \geq (>) \ \bs{p} \cdot \bbs{x}$. Reny's Theorem 1 shows that a dataset can be rationalized if and only if the data satisfies the generalized axiom of revealed preferences (GARP) which is a basic consistency requirement on the revealed preference relations of the consumer. Further, Reny shows that if $D$ can be rationalized then it can also be rationalized by a utility function $U$ which is quasiconcave and increasing in the sense that $\bs{x} \geq \bs{x}'$ implies $U(\bs{x}) \geq U(\bs{x}')$ and $\bs{x} \gg \bs{x}'$ implies $U(\bs{x}) > U(\bs{x}')$.

There are several major differences between Reny's result and our Theorem \ref{theorem:pop-rat-quasilinear}. Most obviously, we are engaged in testing different models. The model investigated by Reny demands less of the data than the models we investigate. Second, our datasets are different: Reny's is a subset of $\eucp \times \eucp$ and ours is a \emph{probability measure} on $\eucpp \times X$. This additional structure is important for our purposes because (i) the concept of an additive rationalization requires that we add up utility across observations and (ii) the concept of the total money pump requires that we add up extracted surplus. 

Another distinction between our results concerns which properties of the utility function are untestable. Reny shows that whenever a dataset can be rationalized then it can be rationalized by a quasiconcave and increasing utility function. In contrast, our result shows that when testing the quasilinear or additive model it happens that requiring utility to be weakly increasing and concave brings no new testable restrictions.

The concept of cyclical monotonicity and quasilinear rationalization only depend on the \emph{support} of the probability measure $\pi_0$. Based on this observation we can easily prove a version of Reny's Theorem for the model of quasilinear utility maximization. 

Let us start with a Reny-style dataset $D \subseteq \eucpp \times X$ where $(\bs{p}, \bs{x}) \in D$ means that the consumer bought $\bs{x}$ when prices were $\bs{p}$. We say that $D$ satisfies cyclical monotonicity if $D'$ satisfies cyclical monotonicity for each non-empty finite subset $D' \subseteq D$. We say that $D$ is quasilinear rationalized by utility function $U$ if $U( \bs{x} ) - \bs{p} \cdot \bs{x} \geq U( \tbs{x} ) - \bs{p} \cdot \tbs{x}$ for all $(\bs{p}, \bs{x}) \in D$ and all $\tbs{x} \in X$ and $U(\bs{x})$ is finite for some $(\bs{p}, \bs{x}) \in D$. Here is a version of Reny's Theorem for the model of quasilinear utility maximization. 
\begin{proposition} \label{prop:reny}
	Let $D \subseteq \eucpp \times X$ be a non-empty dataset. The following are equivalent.
	\begin{enumerate}
		\item $D$ satisfies cyclical monotonicity.
		\item $D$ can be quasilinear rationalized by some $U$.
		\item $D$ can be quasilinear rationalized by some standard $U$.
	\end{enumerate}
\end{proposition}
\noindent This proposition is a simple corollary of Theorem \ref{theorem:pop-rat-quasilinear}. 

\subsection{$\text{TMP}_0$ as a measure of rationality} \label{sec:tmp-pop}

Theorem \ref{theorem:tmp} shows that $\text{TMP} = A = Q$. A version of this result holds for $\text{TMP}_0$. 

The quasilinear utility inefficiency $Q_0$ is the average difference between the maximum amount of quasilinear utility achievable and the quasilinear utility actually achieved\footnote{Note that the function $V(\mbp) = \Big[ \sup_{ \tmbx \in X } U( \tmbx ) - \mbp \cdot \tmbx \Big]$ is upper semi-continuous (as it is the pointwise supremum of continuous functions) and thus there are no measurability issues in our definition of $Q_0$.} 
\begin{equation*} 
	Q_0 = \inf_{U } \left( \E_{\pi_0} \left[ \sup_{ \tmbx \in X } \Big\{ U(\tmbx) - \mbp \cdot \tmbx \Big\} - \Big[ U( \mbx ) -  \mbp \cdot \mbx \Big] \right] \right)
\end{equation*}
where the infimum is over all measurable utility functions $U$ for which  $\E_{\pi_0}[ U(\bs{x}) ]$ is finite. 

For a measurable utility function $U$ let $e_U$ denote the smallest amount of money which must be spent to acquire average utility $\E_{\pi_0}[ U(\bs{x})]$. That is,
\begin{equation*}
	e_U = \inf_{ \pi \in \Pi(\mu_0) } \Big\{ \E_{ \pi }[ \mbp \cdot \mbx ]: \E_{\pi}[ U(\mbx) ] \geq \E_{\pi_0}[ U(\bs{x}) ] \Big\}
\end{equation*}
The additive cost inefficiency $A_0$ is the amount of money which can be saved while still attaining the same average utility
\begin{equation*} \label{eq:pop-A}
	A_0 = \inf_{U} \Big\{ \E_{\pi_0} \big[ \mbp \cdot \mbx \big] - e_U \Big\}
\end{equation*}
where the infimum is over all measurable utility functions $U$ for which  $\E_{\pi_0}[ U(\bs{x}) ]$ is finite.  Here is the infinite data analogue of Theorem \ref{theorem:tmp}. 
\begin{theorem} \label{theorem:pop-tmp-equi}
	For any dataset $\pi_0$ with finite variance and with marginals $\mu_0$ and $\nu_0$
	\begin{equation*} \label{eq:equivalence-pop}
		\text{TMP}_0 = A_0 = Q_0
	\end{equation*}
	Moreover, 
	\begin{enumerate}
		\item there exists a $\pi^* \in \Pi( \mu_0, \nu_0 )$ which achieves the supremum in the definition of $\text{TMP}_0$.
		\item there exists a standard $U$ which achieves the infima in the definitions of $A_0$ and $Q_0$.
		\item a coupling $\pi \in \Pi(\mu_0, \nu_0)$ achieves the supremum in the definition of $\text{TMP}_0$ if and only if $\pi$ satisfies cyclical monotonicity.
	\end{enumerate}
	Suppose further that $\nu_0$ has a positive density on the non-empty connected interior of its support and suppose standard utility functions $U$ and $U'$ both achieve the infimum in the definition of $Q_0$. Then, there exists a number $b \in \mathbb{R}$ so that
	\begin{equation*} 
		U(\bs{x}) = U'( \bs{x} ) + b, \qquad\qquad \text{ for all } \bs{x} \in \supp( \nu_0 )^o
	\end{equation*}
\end{theorem}

Theorem \ref{theorem:pop-tmp-equi} assets that (i) the total money pump, the additive cost inefficiency, and the quasilinear utility inefficiency are all the same and (ii) not only is the infimum in the definition of $Q_0$ achieved by some standard utility function $U$ but in fact this utility function is unique up to translation provided the interior of the support of the probability measure $\nu_0$ is connected. 

The proof of Theorem \ref{theorem:pop-tmp-equi} relies on results from the mathematics of optimal transport. Specifically, we use Kantorovich Duality (see Lemma \ref{lemma:kantorovich} in the appendix) to prove that $\text{TMP}_0 = Q_0$ and we use results from \citet{barrio19} to prove that a standard $U$ which achieves the infimum in the definition of $Q_0$ is unique up to translation. 

\subsection{Asymptotics}

Here we connect the finite dataset measure $\text{TMP}$ with its infinite data analogue $\text{TMP}_0$. Specifically, we show that as the number of observations $T$ grows the finite measure $\text{TMP} / T$ approaches $\text{TMP}_0$. We also show that in a certain sense we are able to asymptotically recover the utility function which achieves the infimum in the definition of $Q_0$.

We consider a dataset which is a stochastic infinite sequence of prices and consumption pairs $D_0 = ( \mbp^t, \mbx^t )_{t \in \mathbb{Z}}$ where $(\bs{p}^t, \bs{x}^t)$ follows distribution $\pi_0$ for each $t$. While the process $D_0$ represents the infinite past and future of prices and consumption we might imagine the researcher only has access to a finite sub-sample $D_T = ( \bs{p}^t, \bs{x}^t )_{t \in \{1,2,\ldots, T\}}$. Using $D_T$ the researcher can estimate $\text{TMP}_0$ (the long run money pump for $\pi_0$) using the total money pump $\text{TMP}_T$ for dataset $D_T$. It turns out that the estimate $\text{TMP}_T/T$ converges to the long run object $\text{TMP}_0$.

\begin{theorem} \label{theorem:tmp-conv}
	Suppose $D_0 = (  \mbp^t, \bs{x}^t )_{t \in \mathbb{Z}}$ is an ergodic stationary stochastic process\footnote{Saying that the stochastic process $D_0$ is stationary means that, for each $T \in \mathbb{N}$, the joint distribution of $( \bs{p}^{1+k}, \bs{x}^{1+k} ), ( \bs{p}^{2+k}, \bs{x}^{2 + k}), \ldots, ( \bs{p}^{T+k}, \bs{x}^{T+k} )$ is the same for all $k$. Saying that $D_0$ is ergodic intuitively means that parts of the data which are far apart (i.e.\ their indices $t$ are far apart) are ``almost'' independent. For a precise definition see Chapter 20 of \citet{klenke-probabilitytheory}. Ergodicity is important because it allows for a version of the strong law of large numbers to hold even when the observations are not independent.} where $( \bs{p}^t, \bs{x}^t)$ has finite variance and let $D_T = ( \bs{p}^t, \bs{x}^t )_{t \in \{1,2,\ldots, T\}}$ for each $T$. Let $\pi_0$ denote the distribution of $(\bs{p}^t, \bs{x}^t)$ and let $\mu_0$ and $\nu_0$ denote the distributions of $\bs{p}^t$ and $\bs{x}^t$, respectively. Let $\text{TMP}_T$ be the total money pump for $D_T$ and let $\text{TMP}_0$ be the total money pump for $\pi_0$. Then, 
	\begin{equation*}
		\dfrac{ \text{TMP}_T }{ T } \to \text{TMP}_0, \qquad\qquad \text{almost surely.}
	\end{equation*}
	Suppose further that $\nu_0$ has a positive density on the non-empty connected interior of its support and let $\bbs{x}$ be some element in $\supp(\nu_0)^0$. Suppose $U$ is a standard utility function which achieves the infimum in the definition of $Q_0$ and, for each $T$, suppose $U_T$ is a standard utility function which achieves the infimum in the definition of $Q$ for dataset $D_T$. Further, suppose that $U(\bbs{x}) = U_T(\bbs{x}) = 0$ for all $T$. Then, almost surely,
	\begin{equation*} 
		U_T(\bs{x}) \to U(\bs{x}), \qquad\qquad \text{ for all } \bs{x} \in \supp( \nu_0 )^o
	\end{equation*}
\end{theorem}
Theorem \ref{theorem:tmp-conv} requires that the stochastic process $D_0$ is ergodic stationary. Naturally, if the observations are IID then this is satisfied. However, we need not limit ourselves to the IID case and in Section \ref{sec:price-misperception} we provide some results which utilize sequences of prices which are ergodic stationary but not IID.

Theorem \ref{theorem:tmp-conv} has two takeaways. First, as noted, the average money pump $\text{TMP}_T / T$ tends to its population level analogue $\text{TMP}_0$ as the number of observations grow to infinity. Second, the standard utility function $U$ which achieves the infimum in the definition of $Q_0$ is asymptotically recovered. A corollary of this is that the ``utility differences'' objects in \eqref{eq:u-diffs} can also be recovered. Specifically, if we let $U_{\ua}^T(\bs{x}^{\diamond}, \bs{x}^{\square}; \mathcal{U}_{st})$ denote the object in \eqref{eq:utility-diffs} for dataset $D_T$ (where $\mathcal{U}_{st}$ is the collection of standard utility functions which achieve the infimum in the definition of $Q$ for dataset $D_T$) and we assume that $\bs{x}^{\diamond}, \bs{x}^{\square} \in \supp(\nu_0)^0$ then, under the assumptions of Theorem \ref{theorem:tmp-conv}, we have that almost surely,
\begin{equation*}
	U_{\ua}^T( \bs{x}^{\diamond}, \bs{x}^{\square}; \mathcal{U}_{st} ) \to U(\bs{x}^{\diamond}) - U( \bs{x}^{\square} )  
\end{equation*}
where $U$ is any standard utility function which achieves the infimum in the definition of $Q_0$. Recall that Theorem \ref{theorem:welfare-direct} shows that $U_{\ua}^T( \bs{x}^{\diamond}, \bs{x}^{\square}; \mathcal{U}_{st} )$ can be calculated by solving a simple linear programming problem. 

The proof of Theorem \ref{theorem:tmp-conv} works by applying the Ergodic Theorem to show that $\pi_T$ converges weakly to $\pi_0$ (where $\pi_T$ is the empirical measure of $D_T$) and manipulating convergence properties of the $2$-Wasserstein metric. The proof that $U_T$ converges relies on optimal transport results from \citet{barrio19}.

An obvious corollary of Theorem \ref{theorem:tmp}, Theorem \ref{theorem:pop-tmp-equi}, and Theorem \ref{theorem:tmp-conv} is that, under the assumptions of Theorem \ref{theorem:tmp-conv} we have, almost surely, $A_T / T \rightarrow A_0$ and  $Q_T / T \rightarrow Q_0$ (where $A_T$ and $Q_T$ are the numbers $A$ and $Q$ for dataset $D_T$). 

The results on the identification and recovery of the utility function in Theorem \ref{theorem:pop-tmp-equi} and \ref{theorem:tmp-conv} require that the support of $\nu_0$ has a non-empty interior. This rules out the important discrete choice case where $X$ is finite. The following covers this omission. 

\begin{proposition} \label{prop:unique-U-finite}
	Suppose $\pi_0$ is a distribution on $\eucpp \times X$ with finite variance and with marginals $\mu_0$ and $\nu_0$. Suppose $\mu_0$ has connected support, $\nu_0$ has finite support, and suppose $U$ and $U'$ are standard utility functions which achieve the infimum in the definition of $Q_0$. Then, there exists a number $b \in \real$ so that
	\begin{equation} \label{eq:U-id}
		U(\bs{x}) = U'(\bs{x}) + b, \qquad \qquad \text{for all } \bs{x} \in \supp(\nu_0)
	\end{equation}
	Suppose further that $D_0 = ( \bs{p}^t, \bs{x}^t )_{t \in \ints}$ is an ergodic stationary stochastic process where $( \bs{p}^t, \bs{x}^t )$ has distribution $\pi_0$. For each $T$, suppose $U_T: X \to \ereal$ is a standard utility function which achieves the infimum in the definition of $Q$ for dataset $D_T = (\bs{p}^t, \bs{x}^t)_{t \in \{1,2,\ldots, T\}}$ and suppose there is some $\bbs{x} \in \supp(\nu_0)$ so that $U(\bbs{x}) = U_T(\bbs{x}) = 0$ for all $T$. Then, almost surely,
	\begin{equation} \label{eq:U-conv2}
		U_T(\bs{x}) \to U( \bs{x} ), \qquad\qquad \text{for all } \bs{x} \in \supp(\nu_0)
	\end{equation}
\end{proposition}
The proposition shows that, if $\nu_0$ has finite support (which holds automatically when $X$ is finite) then the standard utility function $U$ which achieves the infimum in the definition of $Q_0$ is identified up to translation by the data $\pi_0$ and can be asymptotically recovered using finite data $D_T = (\bs{p}^t, \bs{x}^t)_{t \in \{1,2,\ldots, T\}}$.

\section{A model of price misperception} \label{sec:price-misperception}

We consider a consumer who maximizes a quasilinear utility function $U_0$ using incorrect mis-perceived prices. So, in each period $t$ the consumer faces price vector $\bs{p}^t \in \eucpp$ but mis-perceives the price vector as $\tbs{p}^t \in \eucpp$. The consumer then selects a bundle $\bs{x}^t$ so as to maximize quasilinear utility given the perceived price vector $\tbs{p}^t$:
\begin{equation} \label{eq:price-misperception}
	U_0( \bs{x}^t ) - \tbs{p}^t \cdot \bs{x}^t = \max_{ \bs{x} \in X } \Big\{ U_0(\bs{x}) - \tbs{p}^t \cdot \bs{x} \Big\}
\end{equation}
A dataset for this type of consumer is a stochastic process $( \tbs{p}^t, \bs{p}^t, \bs{x}^t )_{t \in \ints}$ where $\bs{p}^t$ is the true price vector, $\tbs{p}^t$ is the perceived price vector, and $\bs{x}^t$ is the consumption bundle which was purchased so as to satisfy \eqref{eq:price-misperception}.\footnote{Similar models of price misperception appear in \citet{chetty-looney-kroft09}, \citet{gabaix14}, and \citet{quah-tserenjigmid25}.} Naturally, the researcher only observes real prices and consumption: $(\bs{p}^t, \bs{x}^t)$.

We state the assumptions in our model more precisely.
\begin{assumption} \label{assumption:price-misperception}
	$( \tbs{p}^t, \bs{p}^t, \bs{x}^t )_{t \in \ints}$ is a stationary ergodic process where $( \tbs{p}^t, \bs{p}^t, \bs{x}^t )$ has finite variance. The consumer has a standard utility function $U_0$ so that, in each period $t$, \eqref{eq:price-misperception} holds. We also assume that $\E_{\pi_0}[ U_0(\bs{x}) ]$ is finite.
\end{assumption}

We connect our model of price misperception to quasilinear utility inefficiency $Q_0$.
\begin{proposition} \label{prop:price-misperception}
	Suppose $( \tbs{p}^t, \bs{p}^t, \bs{x}^t )_{t \in \ints}$ and $U_0$ satisfy Assumption \ref{assumption:price-misperception}. Let $\pi_0$ denote the distribution of $( \bs{p}^t, \bs{x}^t )$ and let $\ti{\mu}_0$, $\mu_0$, and $\nu_0$ denote the distributions of $\tbs{p}^t$, $\bs{p}^t$, and $\bs{x}^t$, respectively. Assume $\tbs{p}^t$ and $\bs{p}^t$ follow the same distribution (i.e.\ $\ti{\mu}_0 = \mu_0$). Then, $U_0$ achieves the infimum in the definition of $Q_0$.
\end{proposition}
The main assumptions in Proposition \ref{prop:price-misperception} is that (i) the consumer best-responds to the perceived prices $\tbs{p}^t$ in the sense of \eqref{eq:price-misperception} and (ii) the perceived prices $\tbs{p}^t$ have the same distribution as the real prices $\bs{p}^t$. 

Let $\bar{Q}(U)$ be the quasilinear utility inefficiency for utility function $U$:
\begin{equation*}
	\bar{Q}(U) = \E_{\pi_0} \bigg[ \sup_{\tbs{x} \in X} \Big\{ U(\tbs{x}) - \bs{p} \cdot \tbs{x} \Big\} - \Big[ U(\bs{x}) - \bs{p} \cdot \bs{x} \Big] \bigg]
\end{equation*}

Proposition \ref{prop:price-misperception} shows that, under certain assumptions, the true utility function $U_0$ achieves the infimum in the definition of $Q_0$ (i.e.\ $Q_0 = \bar{Q}(U_0)$). This is significant because Theorem \ref{theorem:tmp}, Theorem \ref{theorem:tmp-conv}, and Proposition \ref{prop:unique-U-finite} now show that (i) the solution to the linear program in \eqref{eq:lp} converges to the true quasilinear utility inefficiency $\bar{Q}(U_0)$ and (ii) for any two bundles $\bs{x}^{\di}$ and $\bs{x}^{\sq} \in \supp(\nu_0)^o$ the solution to the linear program in \eqref{eq:lp-ineqs-direct} converges to the true utility difference $U_0(\bs{x}^{\di}) - U_0(\bs{x}^{\sq})$ as $T$ becomes large. In other words, under the assumptions of Proposition \ref{prop:price-misperception} we are able to recover, by solving linear programming problems, the true amount of quasilinear utility wasted $\bar{Q}(U_0)$ and the true utility function $U_0$ (up to translation).

Proposition \ref{prop:price-misperception} requires that real prices and perceived prices have the same distribution. We shall next provide several examples where this assumption holds. First, we introduce an assumption on the real prices $\{ \bs{p}^t \}_{t \in \ints}$.

\begin{assumption} \label{assumption:ar1}
	The prices $\{ \bs{p}^t \}_{t \in \ints}$ follow an AR(1) process:
	\begin{equation} \label{eq:ar1}
		\bs{p}^{t} = (1-\alpha) \bbs{p} + \alpha \bs{p}^{t-1} + \bs{\varepsilon}^t 
	\end{equation}
	where $\alpha \in [0,1)$, $\bbs{p}$ is a constant vector, and $\{ \bs{\varepsilon}^t \}_{t \in \mathbb{Z}}$ are IID mean-zero shocks where $\bs{\varepsilon}^t \indep \bs{p}^{t-k}$ for all $k > 0$.
\end{assumption}
Equation \eqref{eq:ar1} says that current prices $\bs{p}^t$ are a convex combination of the previous period's prices $\bs{p}^{t-1}$ and a constant vector $\bbs{p}$ plus an IID shock $\bs{\varepsilon}^t$. It is easy to see that $\E[ \bs{p}^t ] = \bbs{p}$ under \eqref{eq:ar1}. Assumption \ref{assumption:ar1} will allow us to provide interpretable expressions for $\bar{Q}(U_0)$ in the coming examples.

\subsubsection{Lagged prices}

Suppose the consumer is not able to process new price information immediately. Instead, the consumer is only aware of the price vector lagged by some number $k \in \{0,1,2,\ldots\}$. In other words, $\tmbp^t = \mbp^{t-k}$ for each $t$. This clearly implies $\tbs{p}^t$ and $\bs{p}^t$ have the same distribution provided $D_0 = ( \tbs{p}^t, \bs{p}^t, \bs{x}^t )_{t \in \ints}$ is stationary. We write $\bs{p}^t \sim \tbs{p}^{t}$ to indicate that these random vectors have the same distribution and we write $(\bs{x}^t \indep \bs{p}^{t-s}) \ | \ \bs{p}^{t-k}$ to indicate that $\bs{x}^t$ and $\bs{p}^{t-s}$ are independent conditional on $\bs{p}^{t-k}$.
\begin{proposition} \label{prop:lagged-prices}
	Suppose $( \tbs{p}^t, \bs{p}^t, \bs{x}^t )_{t \in \ints}$ and $U_0$ satisfy Assumption \ref{assumption:price-misperception} and suppose that there is a $k \in \{0,1,2,\ldots\}$ so that $\tbs{p}^t = \bs{p}^{t-k}$ for all $t$. Then, $\bs{p}^t \sim \tbs{p}^t$. Suppose additionally that prices follow the AR(1) process of Assumption \ref{assumption:ar1} and suppose $U_0$ is strictly concave. Let $S_0$ be defined by
	\begin{equation} \label{eq:S}
		S_0 = \max_{ \pi \in \Pi(\mu_0, \nu_0) } \Big\{ \E_{\pi}[ \bs{p} ] \cdot \E_{ \pi }[ \bs{x} ] - \E_{\pi}[ \bs{p} \cdot \bs{x} ] \Big\} 
	\end{equation}
	where $\mu_0$ and $\nu_0$ are the distributions of $\bs{p}^t$ and $\bs{x}^t$, respectively. Then,\footnote{The formula \eqref{eq:lagged-prices} uses the convention that $0^0 = 1$.}
	\begin{equation} \label{eq:lagged-prices}
		\bar{Q}(U_0) = (1- \alpha^k) S_0
	\end{equation}
\end{proposition}
Proposition \ref{prop:lagged-prices} makes two claims. First, it shows that $\bs{p}^t \sim \tbs{p}^t$. Thus, under the assumptions of this proposition the conclusion of Proposition \ref{prop:price-misperception} holds and we are able to obtain meaningful estimates of $\bar{Q}(U_0)$ and $U_0$. Second, equation \eqref{eq:lagged-prices} provides a formula for $\bar{Q}(U_0)$ under the lagged perceived prices model and under the additional assumption that prices follow an AR(1) process. 

From \eqref{eq:S} we see that an upper bound for $\bar{Q}(U_0)$ is $S_0$.\footnote{Recall that $\bar{Q}(U_0) = Q_0 = \text{TMP}_0$ and so \eqref{eq:lagged-prices} also provides an upper bound for $\text{TMP}_0$.} The number $S_0$ can be interpreted as a measure of the total possible surplus which can be achieved by best-responding to prices. To elaborate, suppose that $\tbs{x}^t$ is the demand of a naive consumer who acquires $\nu_0$ but does not respond to prices at all and suppose that $\bbs{x}^t$ is the demand of a sophisticated consumer who acquires $\nu_0$ in the cheapest way possible. That is, (i) $\tbs{x}^t \sim \bs{x}^t$, (ii) $\tbs{x}^t \indep \bs{p}^t$, (iii) $\bbs{x}^t \sim \bs{x}^t$, and (iv) $\E[\bs{p}^t \cdot \bbs{x}^t] = \min_{\pi \in \Pi(\mu_0, \nu_0)} \E_{\pi} [ \bs{p} \cdot \bs{x} ]$. It is easy to see that $S_0 = \E[ \bs{p}^t \cdot ( \tbs{x}^t - \bbs{x}^t ) ]$ and so $S_0$ is the difference in expenditure between how much the naive consumer spends and how much the sophisticated consumer spends. That is, $S_0$ is the surplus which the consumer can acquire through sophistication.  

While $S_0$ is the total surplus available we see that $\bar{Q}(U_0) < S_0$ when $\alpha \neq 0$. This is because when $\alpha \neq 0$ then the prices in the previous periods are correlated with the prices in the current period. Thus, even though the consumer is best-responding to out-dated prices these prices still carry valuable information about current period prices and so the amount of wasted utility is less than the total surplus $S_0$. When $\alpha = 0$ then prices are IID and so there is no more advantage to best-responding to an antiquated price then there is to ignoring prices all-together. Thus, when $\alpha=0$ we have $\bar{Q}(U_0) = S_0$.

The role of $k$ in \eqref{eq:lagged-prices} is rather intuitive. If $k = 0$ (and so $\bs{p}^t = \tbs{p}^t$) then no utility is wasted: $\bar{Q}(U_0) = 0$. On the other hand, as $k$ increases (and so $\tbs{p}^t$ is more-and-more out of date) the amount of utility wasted increases. In the limit the wasted utility becomes equal to the total surplus $S_0$. 

\subsubsection{Inattention}

Suppose the consumer can fail to pay attention to how prices change. If the consumer pays attention then she updates prices so that $\tbs{p}^t = \bs{p}^t$. If she fails to pay attention then her prices remain un-updated $\tbs{p}^t = \tbs{p}^{t-1}$. More formally, let $\{ a_t \}_{t \in \ints}$ denote binary random variables where $a_t = 1$ means that the consumer updated prices and $a_t = 0$ means that the consumer failed to update prices. We then have
\begin{equation} \label{eq:price-updates}
	\tbs{p}^t = \begin{cases}
		\tbs{p}^{t-1}, & \qquad \text{if } a_t = 0 \\
		\bs{p}^t, & \qquad \text{if } a_t = 1
	\end{cases}
\end{equation}
The following provides a result on this inattention model.
\begin{proposition} \label{prop:inattention}
	Suppose $( \tbs{p}^t, \bs{p}^t, \bs{x}^t )_{t \in \ints}$ and $U_0$ satisfy Assumption \ref{assumption:price-misperception} and that mental prices follow the inattentive model of \eqref{eq:price-updates} where $\{ a_t \}_{t \in \ints}$ is an IID sequence of binary random variables which are independent of prices $\{ \bs{p}^t \}_{t \in \mathbb{Z}}$ and which satisfy $P( a_t = 0 ) = \beta \in [0,1)$. Then $\bs{p}^t \sim \tbs{p}^t$. Suppose additionally that prices follow the AR(1) process of Assumption \ref{assumption:ar1} and that $U_0$ is strictly concave. Letting $S_0$ be defined by \eqref{eq:S} we have
	\begin{equation} \label{eq:tmp-inattention}
		\bar{Q}(U_0) = \beta \left( \dfrac{ 1- \alpha }{ 1 - \alpha \beta } \right) S_0
	\end{equation}
\end{proposition}

The first conclusion of Proposition \ref{prop:inattention} is that $\bs{p}^t \sim \tbs{p}^t$ under our inattention model. This is rather obvious considering $\tbs{p}^t$ is a mixture of $\bs{p}^t, \bs{p}^{t-1}, \bs{p}^{t-2}, \ldots$. The second result is a formula for $\bar{Q}(U_0)$ which holds under the assumption that prices follow an AR(1) process. The number $S_0$ in \eqref{eq:tmp-inattention} has the same interpretation as in \eqref{eq:lagged-prices}. From \eqref{eq:tmp-inattention} it is easy to see that $S_0$ is again an upper bound on $\bar{Q}(U_0)$ (and thus $S_0$ is also an upper bound on $\text{TMP}_0$).

The parameter $\beta$ is the probability with which the consumer fails to update the mental prices. Naturally, when $\beta = 0$ we see that $\bar{Q}(U_0) = 0$ and as $\beta$ increases so too does $\bar{Q}(U_0)$. As $\beta$ approaches 1 the wasted utility $\bar{Q}(U_0)$ converges to the total surplus $S_0$. 

\subsubsection{Forecasting with errors}

As in our lagged price model, suppose that the consumer does not process all the new price information instantly. The consumer observes some out of date price $\bs{p}^{t-k}$ and instead of taking $\bs{p}^{t-k}$ as the perceived price the consumer creates a mental forecast of $\bs{p}^t$ by taking a random draw from the distribution of $\bs{p}^t$ conditional on $\bs{p}^{t-k}$. The following proposition provides a formula for $\bar{Q}(U_0)$ under the assumptions of this model.  
\begin{proposition} \label{prop:false-forecast}
	Suppose $( \tbs{p}^t, \bs{p}^t, \bs{x}^t )_{t \in \ints}$ and $U_0$ satisfy Assumption \ref{assumption:price-misperception} and suppose that $\tbs{p}^t$ and $\bs{p}^t$ are IID conditional on $\bs{p}^{t-k}$ where $k \in \mathbb{N}$. Then, $\bs{p}^t \sim \tbs{p}^t$. Suppose further that prices follow the AR(1) process of Assumption \ref{assumption:ar1} and that $U_0$ is strictly concave. Letting $S_0$ be defined by \eqref{eq:S} we have
	\begin{equation} \label{eq:false-forecast}
		\bar{Q}(U_0) = (1- \alpha^k) S_0 + \alpha^k \bs{E}[ ( \bs{p}^{t-k} - \tbs{p}^t ) \cdot \bs{x}^t ] 
	\end{equation}
	Also, $\bar{Q}(U_0)$ is increasing in $k$ and $\bar{Q}(U_0) \to S_0$ as $k \to \infty$.
\end{proposition}
Comparing the expressions for $\bar{Q}(U_0)$ in \eqref{eq:lagged-prices} and \eqref{eq:false-forecast} we see that the wasted utility $\bar{Q}(U_0)$ in the lagged prices model always lies below the wasted utility in the forecasting with errors model (assuming the same $U_0$ and the same values of $\alpha$ and $k$ in both models). 

This ``forecasting with errors'' model of consumer behavior is similar to the S(1) decision making process in \citet{rubinstein-spiegler08}.\footnote{Similar decision making procedures appear in \citet{obsorne-rubinstein98}, \citet{spiegler06}, and \citet{spiegler06-2}.} In Rubinstein and Spiegler investors decide whether or not to purchase an asset based on their forecasts of the asset's future price. Instead of performing some sort of Bayesian analysis for the forecast the investors simply pick a random draw from the distribution of the future price conditional on the current prices. In our model, the consumer is too busy to personally inspect prices in period $t$ and so instead the consumer makes a forecast by taking a random draw from the distribution of current prices conditional on lagged prices. 

\section{Conclusion} \label{sec:conclusion}

This paper has introduced the total money pump (TMP) and has shown that the TMP is always equal to a tight lower bound on wasted quasilinear utility which is also equal to a tight lower bound on the cost inefficiency displayed by an additive utility maximizer. This connection between the TMP and measures of bounded rationality has been shown to persist beyond the typical finite data relm of revealed preference theory and in particular the relationship holds when a dataset is an arbitrary joint distribution over prices and consumption. We have seen that as the number of observations $(\mbp^t, \mbx^t)$ increases we are able to asymptotically recover a long run money pump and the true utility function of the consumer provided the consumer acts as a quasilinear utility maximizer who misperceives prices in a particular structure way.

\bibliography{rpbib}
\bibliographystyle{myplainnat}

\vspace{40pt}

\appendix

\noindent {\Huge Appendix}

\section{Characterizations of the argmin in $Q$} \label{sec:argmin-char}

In the welfare section (Section \ref{sec:welfare}) we mentioned that there are alternative ways to characterize the collection of utility functions which achieve the infimum in the definition of $Q$. We present and interpret these alternative characterizations here.

We interpret a permutation $\sigma$ of $\{1,2,\ldots, T\}$ as a way of correcting the choices of the consumer. Specifically, suppose that a consumer made choices $D = ( \bs{p}^t, \bs{x}^t )_{t \leq T}$ which do not satisfy cyclical monotonicity. We might use a permutation $\sigma$ to suggest a corrected pattern of behavior $D_{\sigma} = ( \bs{p}^t, \bs{x}^{\sigma(t)} )_{t \leq T}$ for the consumer. Note that the corrected data $D_{\sigma}$ involves purchasing the exact same distribution of consumption as in the original data but allows for potential cost savings. We might take $\mathcal{U}$ in \eqref{eq:utility-diffs} to be all utility functions which quasilinear rationalize the corrected dataset $D_{\sigma} = ( \bs{p}^t, \bs{x}^{\sigma(t)} )$ for some permutation $\sigma$. We might additionally require that such a $\sigma$ achieves the maximum in the definition of $\text{TMP}$. It turns out that these approaches are equivalent to considering the utility functions which achieve the infimum in the definition of $Q$.
\begin{proposition} \label{prop:argmin-char}
	Let $D = (\mbp^t, \mbx^t)_{t \leq T}$ be a dataset and let $U: X \to \ereal$ be a utility function. The following are equivalent.
	\begin{enumerate}
		\item $U$ achieves the infimum in the definition of $Q$.
		\item $U$ quasilinear rationalizes $D_{\sigma} = ( \bs{p}^t, \bs{x}^{\sigma(t)} )_{t \leq T}$ for some permutation $\sigma$.
		\item $U$ quasilinear rationalizes $D_{\sigma} = ( \bs{p}^t, \bs{x}^{\sigma(t)} )_{t \leq T}$ for every $\sigma$ which achieves the maximum in the definition of $\text{TMP}$.
	\end{enumerate}
\end{proposition}
\begin{proof}
	This is an immediate corollary of Lemma \ref{lemma:Q-A-T2} and Theorem \ref{theorem:tmp}.
\end{proof}

A similar result holds for the infinite data case.

\begin{proposition} \label{prop:argmin-char-pop}
	Let $\pi_0$ be a distribution on $\eucpp \times X$ with finite second moments and with marginals $\mu_0$ and $\nu_0$, respectively. Let $U: X \to \ereal$ be a measurable utility function. The following are equivalent.
	\begin{enumerate}
		\item $U$ achieves the infimum in the definition of $Q_0$.
		\item $U$ quasilinear rationalizes some $\pi \in \Pi(\mu_0, \nu_0)$.
		\item $U$ quasilinear rationalizes every $\pi \in \Pi(\mu_0,\nu_0)$ which achieves the supremum in the definition of $\text{TMP}_0$.
	\end{enumerate}
\end{proposition}
\begin{proof}
	This is an immediate corollary of Lemma \ref{lemma:Q-A-T} and Theorem \ref{theorem:pop-tmp-equi}.
\end{proof}

\section{Proofs}

\subsection{Section \ref{sec:tmp} Proofs}

\noindent \textbf{Theorem \ref{theorem:rat-quasilinear}.}

\begin{proof}[Proof of Theorem \ref{theorem:rat-quasilinear}.]
	Let $\sigma$ be some permutation of $\{1,2,\ldots, T\}$ and note that if $D$ is quasilinear rationalized by $U$ then $\sum_{t=1}^T U(\mbx^t) - \mbp^t \cdot \mbx^t \geq \sum_{t=1}^T U(\mbx^{\sigma(t)}) - \mbp^t \cdot \mbx^{\sigma(t)}$ which, after rearranging, yields $\sum_{t=1}^T \mbp^t \cdot ( \mbx^t - \mbx^{\sigma(t)} ) \leq 0$ and thus item 5 implies item 1. 
	
	Next, suppose $D$ satisfies cyclical monotonicity. Let $U: X \rightarrow \mathbb{R}$ be defined by
	\begin{equation} \label{eq:U-ql-def}
		U(\mbx) = \inf \left( \mbp^{t_K} \cdot (\mbx - \mbx^{t_K} ) + \sum_{k=1}^{K-1} \mbp^{t_k} \cdot ( \mbx^{t_{k+1}} - \mbx^{t_k} ) \right)
	\end{equation}
	where the infimum is taken over all finite sequences $t_1, t_2, \ldots, t_K$. As $D$ satisfies cyclical monotonicity it is easy to see that the infimum in \eqref{eq:U-ql-def} is always attained by some sequence $t_1,t_2, \ldots, t_K$ with at most $T$ elements. As such, $U$ is the pointwise infimum of finitely many well-behaved functions and is thus well-behaved. Let $t \in \{{1,2,\ldots, T}\}$ and let $\ti{t}_1, \ti{t}_2, \ldots, \ti{t}_K$ be the sequence which attains the infimum in \eqref{eq:U-ql-def} for $U(\mbx^{t})$. For any $\mbx \in X$ we have
	\begin{equation*}
		U( \mbx^{t} ) + \mbp^t \cdot ( \mbx - \mbx^t ) = \mbp^t \cdot ( \mbx - \mbx^t ) + \mbp^{ \ti{t}_K} \cdot ( \mbx^t - \mbx^{\ti{t}_K} ) + \sum_{k=1}^{K-1} \mbp^{\ti{t}_k} \cdot ( \mbx^{\ti{t}_{k+1}} - \mbx^{\ti{t}_k} ) \geq U(\mbx)
	\end{equation*}
	where the final inequality follows from the definition of $U$. Rearranging the previous inequality gives $U(\mbx^t) - \mbp^t \cdot \mbx^t \geq U(\mbx) - \mbp^t \cdot \mbx$ and so $D$ is quasilinear rationalized by $U$. Let $(\tmbx^1,\tmbx^2,\ldots, \tmbx^T)$ satisfy $\sum_{t=1}^T \mbp^t \cdot \mbx^t \geq (>) \ \mbp^t \cdot \tmbx^t$. Then, as $D$ is quasilinear rationalized by $U$ we see $\sum_{t=1}^T U(\mbx^t) \geq \sum_{t=1}^T U(\tmbx^t) + \mbp^t \cdot ( \mbx^t - \tmbx^t ) \geq (>) \ \sum_{t=1}^T U(\tmbx^t)$ and so $D$ is additively rationalized by $U$. We have just shown that item 2 implies items 4 and 6. It is easy to see that items 1 and 2 are equivalent and that item 4 implies item 3 and item 6 implies item 5 and so the proof is complete.
\end{proof}

\noindent \textbf{Theorem \ref{theorem:tmp}.}

\noindent For any $U: X \to \ereal$ for which $U(\bs{x}^t)$ is finite for all $t$ let
\begin{equation} \label{eq:Q-bar2}
	\bar{Q}(U) = \sum_{t=1}^T \left[ \sup_{ \tbs{x} \in X } \Big\{ U(\tbs{x}) - \bs{p}^t \cdot \tbs{x} \Big\} \right] - \Big[ U(\bs{x}^t) - \bs{p}^t \cdot \bs{x}^t \Big]
\end{equation}	
and 
\begin{equation} \label{eq:A-bar2}
	\bar{A}(U) = \sum_{t=1}^T \bs{p}^t \cdot \bs{x}^t - e_U
\end{equation}
Note that $Q = \inf_U \bar{Q}(U)$ and $A = \inf_U \bar{A}(U)$ where the infima are over all utility functions for which $U(\bs{x}^t)$ is finite for all $t$. 
\begin{lemma} \label{lemma:Q-A-T2}
	Let $D = (\bs{p}^t, \bs{x})_{t \leq T}$ be a dataset. For any $U: X \to \ereal$ for which $U(\bs{x}^t)$ is finite for all $t$ we have
	\begin{equation} \label{eq:Q-A-T2}
		\bar{Q}(U) \geq \bar{A}(U) \geq \text{TMP}
	\end{equation}
\end{lemma}
\begin{proof}
	For any $U: X \to \ereal$ for which $U(\bs{x}^t)$ is finite for all $t$ and any $( \tbs{x}^1, \tbs{x}^2, \ldots, \tbs{x}^T )$ satisfying $\sum_{t=1}^T U(\tbs{x}^t) = \sum_{t=1}^T U(\bs{x}^t)$ we have
	\begin{equation*}
		Q(U) \geq \sum_{t=1}^T U( \tbs{x}^t ) - \bs{p}^t \cdot \tbs{x}^t - U(\bs{x}^t) - \bs{p}^t \cdot \bs{x}^t \geq \sum_{t=1}^T \bs{p}^t \cdot ( \bs{x}^t - \tbs{x}^t ) 
	\end{equation*}
	and so $Q(U) \geq A(U)$. 
	
	For any $U: X \to \ereal$ for which $U(\bs{x}^t)$ is finite for all $t$ and any permutation $\sigma: \{1,2,\ldots, T \} \to \{1,2,\ldots, T\}$ we have
	\begin{equation*}
		A(U) = \sum_{t=1}^T \bs{p}^t \cdot \bs{x}^t - e_U \geq \sum_{t=1}^T \bs{p}^t \cdot (\bs{x}^t - \bs{x}^{\sigma(t)} )
	\end{equation*}
	and so $A(U) \geq \text{TMP}$. 
\end{proof}

\begin{proof}[Proof of Theorem \ref{theorem:tmp}.]
	From Lemma \ref{lemma:Q-A-T2} we know that $\bar{Q}(U) \geq \bar{A}(U) \geq \text{TMP}$. We will show that (i) $Q \geq \bar{\varepsilon}$, (ii) $\bar{\varepsilon} \geq \text{TMP}$, and (iii) there exists a well-behaved $U$ which satisfies $\bar{Q}(U) = \text{TMP}$ which will complete the proof.
	
	To show that $Q \geq \bar{\varepsilon}$ let $U$ be any utility function for which $U(\bs{x}^t)$ is finite for all $t$ and define $u_t = U(\mbx^t)$ and $\varepsilon_t = \sup_{s} \big(  u_s - \mbp^t \cdot \mbx^s - (u_t - \mbp^t \cdot \mbx^t) \big)$. It is easy to see that these numbers $u_t$ and $\varepsilon_t$ constitute a feasible solution to \eqref{eq:lp} and further that $\sum_{t=1}^T \varepsilon_t$ lies below the quantity $$\sum_{t=1}^T \sup_{\mbx \in X} \left[ U(\mbx) - \mbp^t \cdot \mbx \right] - \left[ U(\mbx^t) - \mbp^t \cdot \mbx^t \right]$$
	and so indeed $Q \geq \bar{\varepsilon}$. 
	
	To see that $\bar{\varepsilon} \geq \text{TMP}$ let $(u_1,u_2,\ldots,u_T)$ and $(\varepsilon_1, \varepsilon_2, \ldots, \varepsilon_T)$ constitute a feasible solution to \eqref{eq:lp} and let $\sigma$ be any permutation of $\{1,2,\ldots, T\}$. Using the constraint inequalities in \eqref{eq:lp} we see
	\begin{equation*}
		\sum_{t=1}^T \mbp^t \cdot ( \mbx^{t} - \mbx^{\sigma(t)} ) = \sum_{t=1}^T \left[ u_{\sigma(t)} - u_t + \mbp^t \cdot ( \mbx^{t} - \mbx^{\sigma(t)} ) \right] \leq \sum_{t=1}^T \varepsilon_t = \bar{\varepsilon}
	\end{equation*}
	and so indeed $\bar{\varepsilon} \geq \text{TMP}$.
	
	Next, let $\sigma$ be a permutation which achieves the supremum in the definition of the $\text{TMP}$. Let $D_{\sigma} = (\mbp^t, \mbx^{\sigma(t)})_{t \leq T}$. We claim that $D_{\sigma}$ satisfies cyclical monotonicity. For a contradiction suppose that this is not the case and thus there exists some permutation $\sigma'$ which ``money pumps'' the dataset $D_{\sigma}$ in the sense that $\sum_{t=1}^T \mbp^{ t } \cdot ( \mbx^{\sigma(t)} - \mbx^{\sigma'( \sigma(t) )} ) > 0$. We have
	\begin{equation*}
		0 < \sum_{t=1}^T \mbp^{ t } \cdot ( \mbx^{\sigma(t)} - \mbx^{\sigma'( \sigma(t) )} ) = \sum_{t=1}^T \mbp^t \cdot ( \mbx^t - \mbx^{\sigma'(\sigma(t))} ) - \sum_{t=1}^T \mbp^t \cdot ( \mbx^t - \mbx^{\sigma(t)} )
	\end{equation*}
	After rearranging we see that $\sum_{t=1}^T \mbp^t \cdot (\mbx^{\sigma'( \sigma(t) )} - \mbx^t) > \sum_{t=1}^T \mbp^t \cdot (  \mbx^{\sigma(t)} - \mbx^t)$ which contradicts the assumption that $\sigma$ achieves the supremum in the definition of the $\text{TMP}$. Having achieved a contradiction we conclude that indeed $D_{\sigma}$ satisfies cyclical monotonicity. 
	
	As $D_{\sigma}$ satisfies cyclical monotonicity we may appeal to Theorem \ref{theorem:rat-quasilinear} to see that there exists a well-behaved $U$ which quasilinear rationalizes $D_{\sigma}$. As $U$ quasilinear rationalizes $D_{\sigma}$ we have
	\begin{IEEEeqnarray*}{rCl}
		\bar{Q}(U) & = & \sum_{t=1}^T [ U(\mbx^{\sigma(t)}) - \mbp^t \cdot \mbx^{\sigma(t)} ] - [ U(\mbx^t) - \mbp^t \cdot \mbx^t ] = \sum_{t=1}^T \bs{p}^t \cdot ( \bs{x}^t - \bs{x}^{\sigma(t)} ) = \text{TMP}
	\end{IEEEeqnarray*}
	which completes the proof.
\end{proof}

\noindent \textbf{Theorem \ref{theorem:welfare-direct}.}

\begin{proof}[Proof of Theorem \ref{theorem:welfare-direct}.]
	We refer to the linear programming problem defined in the statement of this theorem as the LP. Throughout we let $\sigma$ be a permutation of $\{1,2,\ldots, T\}$ which achieves the infimum in the definition of $\text{TMP}$. We extend $\sigma$ by taking $\sigma(\diamond) = \diamond$ and $\sigma( \square ) = \square$. Also, we let $\Delta^{\mathcal{T}}$ denote the collection of vectors in $\mathbb{R}^{T+2}$ with non-negative entries which sum to 1. That is, $(\alpha_{t})_{t \in \mathcal{T}} \in \Delta^{\mathcal{T}}$ satisfies $\alpha_t \geq 0$ for all $t \in \mathcal{T}$ and $\sum_{t \in \mathcal{T}} \alpha_t = 1$. We break the proof into steps. 
	
	\begin{step} \label{step:feasible-solution}
		The LP has a feasible solution: $( (u_t, \varepsilon_t)_{t \in \mathcal{T}}, \bs{p}^{\diamond}, \bs{p}^{\square} )$ with $\bs{p}^{\diamond}, \bs{p}^{\square} \in \eucpp$. 
	\end{step}
	\begin{proof}
		Let $(u^t, \varepsilon^t)_{t \in \{1,2,\ldots, T\} }$ be a feasible solution to the linear programming problem in \eqref{eq:lp}. For each $\star \in \{ \diamond, \square \}$ let $t^{\star}$ be any observation number which achieves the infimum in
		\begin{equation*}
			\inf_{ t \in \{1,2,\ldots, T\} } \ u^t + \bs{p}^t \cdot ( \bs{x}^{\star} - \bs{x}^t ) + \varepsilon^t
		\end{equation*}
		Set $\bs{p}^{\star} = \bs{p}^{t^{\star}}$ and let $u^{\star}$ be defined by
		\begin{equation*}
			u^{\star} = u^{t^{\star}} + \bs{p}^{t^{\star}} \cdot ( \bs{x}^{\star} - \bs{x}^{t^{\star}} ) + \varepsilon^{t^{\star}}
		\end{equation*}
		Setting $\varepsilon^{\diamond} = \varepsilon^{\square} = 0$ we claim we have found a feasible solution. So, we must verify that \eqref{eq:lp-ineqs-direct} holds for all $s,t \in \mathcal{T}$. Evidently these inequalities hold for all $s,t \in \{1,2,\ldots, T\}$. Now, notice that for any $s \in \mathcal{T}$ and $\star \in \{\diamond, \square\}$
		\begin{equation*}
			u^{s} \leq \inf_{ t \in \{1,2,\ldots, T\} } u^t + \bs{p}^t \cdot ( \bs{x}^{s} - \bs{x}^t ) \leq u^{t^{\star}} + \bs{p}^{t^{\star}} \cdot ( \bs{x}^{s} - \bs{x}^{t^{\star}} ) = u^{\star} + \bs{p}^{\star} \cdot ( \bs{x}^{s} - \bs{x}^{\star} )
		\end{equation*}
		The first inequality shows that \eqref{eq:lp-ineqs-direct} holds for $s \in \mathcal{T}$ and $t \in \{1,2,\ldots, T\}$ and the final inequality covers the remaining cases and thus \eqref{eq:lp-ineqs-direct} holds for all $s,t \in \mathcal{T}$.
	\end{proof}
	
	\begin{step} \label{step:U-gen}
		For every feasible solution $( (u^t, \varepsilon^t)_{t \in \mathcal{T}}, \bs{p}^{\diamond}, \bs{p}^{\square} )$ there exists a standard utility function $U: X \to \real$ which achieves the infimum in the definition of $Q$ and which satisfies $U(\bs{x}^{\diamond}) - U(\bs{x}^{\square}) = u^{\diamond} - u^{\square}$. Further, if $\bs{p}^{\di}, \bs{p}^{\sq} \in \eucpp$ then this $U$ can be taken to well-behaved.
	\end{step}
	\begin{proof}
		Let $U(\bs{x}) = \inf_{ t \in \mathcal{T} } \ u^t + \bs{p}^t \cdot (\bs{x} - \bs{x}^t) + \varepsilon^t$. This $U$ is standard and if $\bs{p}^{\di}, \bs{p}^{\sq} \in \eucpp$ then this $U$ is well-behaved. Adding up the inequalities in \eqref{eq:lp-ineqs-direct} (take $s = \sigma(t)$) we see
		\begin{equation*}
			0 = \sum_{ t \in \mathcal{T} } u^{\sigma(t)} - u^t \leq \sum_{t \in \mathcal{T} } \bs{p}^{ t } \cdot ( \bs{x}^{\sigma(t)} - \bs{x}^{t} ) + \varepsilon^t = \text{TMP} - \text{TMP} = 0
		\end{equation*}
		which establishes that each inequality in \eqref{eq:lp-ineqs-direct} holds as an equality when $s = \sigma(t)$. So, for any $s,t \in \mathcal{T}$
		\begin{equation*}
			u^t + \bs{p}^t \cdot ( \bs{x}^{\sigma(t)} - \bs{x}^t ) + \varepsilon^t = u^{ \sigma(t) } \leq u^s + \bs{p}^s \cdot ( \bs{x}^{\sigma(t)} - \bs{x}^s ) + \varepsilon^s
		\end{equation*}
		and so $u^{\sigma(t)} = U( \bs{x}^{\sigma(t)} )$ for all $t \in \mathcal{T}$. In other words, 
		\begin{equation} \label{eq:u-sigma-eq}
			u^{t} = U( \bs{x}^{t} ), \qquad\qquad \forall t \in \mathcal{T}
		\end{equation}
		Using \eqref{eq:u-sigma-eq}, we see that for any $t \in \mathcal{T}$ and $\bs{x} \in X$
		\begin{equation*}
			U( \bs{x} ) - \bs{p}^{t} \cdot \bs{x} \leq u^{\sigma(t)} + \bs{p}^t \cdot ( \bs{x} - \bs{x}^{\sigma(t)} ) - \bs{p}^t \cdot \bs{x} = u^{\sigma(t)} - \bs{p}^t \cdot \bs{x}^{\sigma(t)} = U( \bs{x}^{\sigma(t)} ) - \bs{p}^t \cdot \bs{x}^{\sigma(t)}
		\end{equation*}
		and so $U$ quasilinear rationalizes $D_{\sigma} = ( \bs{p}^t, \bs{x}^{\sigma(t)} )_{t \in \mathcal{T}}$ and so, by Proposition \ref{prop:argmin-char} we see that $U$ achieves the infimum in the definition of $Q$. Further, from \eqref{eq:u-sigma-eq} we see that $U(\bs{x}^{\diamond}) - U(\bs{x}^{\square}) = u^{\diamond} - u^{\square}$.
	\end{proof}
	
	\begin{step} \label{step:wb-ub}
		$U_{\ua}(\mbx^{\di}, \mbx^{\sq}; \mathcal{U}_{wb} ) \geq \hat{U}_{\ua}(\bs{x}^{\di}, \bs{x}^{\sq}) $.
	\end{step}
	\begin{proof}
		The set of feasible solutions to the LP is convex (and non-empty by Step \ref{step:feasible-solution}.) The result now follows by combining Steps \ref{step:feasible-solution} and \ref{step:U-gen} with the observation that $\alpha \bs{p} + (1-\alpha) \bbs{p} \in \eucpp$ when $\bs{p} \in \eucpp$, $\bbs{p} \in \eucp$ and $\alpha \in (0,1)$. 
	\end{proof}
	
	\begin{step} \label{step:support-price}
		Suppose the numbers $( u^t, \varepsilon^t )_{t \in \mathcal{T}}$ satisfy \eqref{eq:lp-ineqs-direct} for $s \in \mathcal{T}$ and $t \in \{1,2,\ldots,T\}$ and suppose $\varepsilon^{\diamond} = \varepsilon^{\square} = 0$. Let $\star \in \{ \diamond, \square \}$. There exists a $\bs{p}^{\star} \in \eucp$ so that \eqref{eq:lp-ineqs-direct} holds for all $s \in \mathcal{T}$ and $t = \star$ if and only if for each $(\alpha_t)_{t\in\mathcal{T}} \in \Delta^{\mathcal{T}}$ we have 
		\begin{equation} \label{eq:support-price}
			\bs{x}^{\star} \geq \sum_{ t \in \mathcal{T} } \alpha_t \bs{x}^t \qquad \implies \qquad u^{\star} \geq \sum_{t \in \mathcal{T}} \alpha_t u^t
		\end{equation}
	\end{step}

	\begin{proof}
		Fix some $\star \in \{\diamond, \square \}$. Define $C \subseteq \mathbb{R}^{L+1}$ by
		\begin{equation*}
			C = \left\{ (\bs{x}, u) \in \mathbb{R}^{L+1}: \exists (\alpha_t)_{t \in \mathcal{T}} \in \Delta^{\mathcal{T}} \text{ s.t.\ } \bs{x} \geq \sum_{t \in\mathcal{T}} \alpha_t \bs{x}^t \ \& \ \sum_{t\in\mathcal{T}} u^t \geq u \right\}
		\end{equation*}
		It is clear that $C$ is a convex polyhedron and so there are vectors $\bs{b}_1, \ldots, \bs{b}_M \in \mathbb{R}^L$ and numbers $a_1, \ldots, a_M \in \real$ and numbers $\beta_1, \ldots, \beta_M \in \real$ so that
		\begin{equation*}
			C = \Big\{ (\bs{x},u) \in \mathbb{R}^{L+1}: a_m u - \bs{b}_m \cdot \bs{x} \leq \beta_m, \qquad \forall m \in \{1,2,\ldots, M \} \Big\}
		\end{equation*}
		Note that if $(\bs{x}, u) \in C$ then $(\bbs{x}, \bar{u}) \in C$ for any $\bbs{x} \geq \bs{x}$ and $\bar{u} \leq u$ which shows that $\bs{b}_m \geq \bs{0}$ and $a_m \geq 0$ for all $m \in \{1,2,\ldots, M\}$. 
		
		Suppose $\bs{p}^{\star} \in \eucp$ is such that \eqref{eq:lp-ineqs-direct} holds for all $s \in \mathcal{T}$ and $t = \star$. Then, for any $(\alpha_t)_{t \in \mathcal{T}}$ such that $\bs{x}^{\star} \geq \sum_{t\in\mathcal{T}} \alpha_t \bs{x}^t$ we have
		\begin{equation*}
			\sum_{t \in \mathcal{T}} \alpha_t u^{t} \leq u^{\star} + \bs{p}^{\star} \cdot \left( \sum_{t\in \mathcal{T}} \bs{x}^t - \bs{x}^{\star} \right) \leq u^{\star}
		\end{equation*}
		and so \eqref{eq:support-price} holds.
		
		Next, suppose \eqref{eq:support-price} holds. From \eqref{eq:support-price} it is clear that $( \bs{x}^{\diamond}, u^{\diamond} + \delta ) \notin C$ for all $\delta > 0$. Thus, there must be some $m^{\diamond} \in \{1,2,\ldots, M\}$ so that $a_{m^{\diamond}} u^{\diamond} - \bs{b}_{m^{\diamond}} \cdot \bs{x}^{\diamond} = \beta_{m^{\diamond}}$ and $a_{m^{\diamond}} (u^{\diamond} - \delta ) - \bs{b}_{m^{\diamond}} \cdot \bs{x}^{\diamond} > \beta_{m^{\diamond}}$ for all $\delta >0$. From this it is clear that $a_{m^{\diamond}} > 0$. Take $\bs{p}^{\diamond} = \bs{b}_{m^{\diamond}} / a_{m^{\diamond}}$ and because $( \bs{x}^{t}, u^t ) \in C$ for all $t \in \mathcal{T}$ we have
		\begin{equation*}
			u^{\diamond} - \bs{p}^{\diamond} \cdot \bs{x}^{\diamond} = \beta_{m^{\diamond}} / a_{m^{\diamond}} \geq u^t - \bs{p}^{\diamond} \cdot \bs{x}^{t}
		\end{equation*}
		Thus, we can find the desired $\bs{p}^{\star}$. 
		
	\end{proof}
	
	\begin{step} \label{step:lp-bounded}
		$\hat{U}_{\ua}(\bs{x}^{\diamond}, \bs{x}^{\square}) < \infty$ if and only if
		\begin{equation} \label{eq:lp-bounded}
			\bs{x}^{\square} \in \conv( \bs{x}^{1},\ldots, \bs{x}^{T}, \bs{x}^{\diamond} ) + \eucp
		\end{equation}
	\end{step}
	\begin{proof}
		Suppose \eqref{eq:lp-bounded} does not hold. Let $( (\bs{x}^t, \varepsilon^t)_{t \in \mathcal{T} \backslash \{\square\}} )$ and $\bs{p}^{\diamond} \in \eucp$ satisfy \eqref{eq:lp-ineqs-direct} for $s,t \in \mathcal{T} \backslash \{ \square \}$ (these objects exist by Step \ref{step:feasible-solution}). Let $u^{\square} \in \real$ be some number which is small enough so that $u^{\square} \leq u^{t} + \bs{p}^t \cdot ( \bs{x}^{\sq} - \bs{x}^t )$ for all $t \in \mathcal{T} \backslash \{\sq\}$. As \eqref{eq:lp-bounded} does not hold we see that the condition on the left side of \eqref{eq:support-price} only holds with $\alpha_{\square} = 1$ and thus Step \ref{step:support-price} guarantees the existence of some $\bs{p}^{\square} \in \eucp$ so that \eqref{eq:lp-ineqs-direct} holds with $t = \sq$ and $s \in \mathcal{T}$. We have thus found a feasible solution to the LP and because $u^{\square}$ can be chosen to be arbitrarily small we see that $\hat{U}_{\ua}(\bs{x}^{\diamond}, \bs{x}^{\square}) = \infty$.
		
		Next, suppose that \eqref{eq:lp-bounded} holds and let $\bs{e}_{\ell} \in \eucp$ denote the vector of all zeros with a one on entry $\ell$. As \eqref{eq:lp-bounded} holds there exists $(\alpha_1, \ldots, \alpha_T, \alpha_{\di}) \in \mathbb{R}^{T+1}$ which sum to $1$ and $(\beta_1, \ldots, \beta_{L}) \in \eucp$ which satisfy
		\begin{equation} \label{eq:x-sq-decom}
			\bs{x}^{\square} = \sum_{t \in \mathcal{T} \backslash \{ \sq \} } \alpha_t \bs{x}^{t} + \sum_{ \ell = 1}^L \beta_{\ell} \bs{e}_{\ell}
		\end{equation} 
		Let $( (u^t,\varepsilon^t)_{t \in \mathcal{T}}, \bs{p}^{\di}, \bs{p}^{\sq} )$ be a feasible solution to the LP. Adding up the inequalities in \eqref{eq:lp-ineqs-direct} (with $t = \square$) and the inequalities $0 \leq \bs{p}^{\square} \cdot \bs{e}_{\ell}$ and using \eqref{eq:x-sq-decom} we have
		\begin{equation*}
			\sum_{ s \in \mathcal{T} \backslash \{\sq\} } \alpha_s u^{s} \leq u^{\square} + \sum_{s \in \mathcal{T} \backslash \{\sq\} } \alpha_s \bs{p}^{\square} \cdot ( \bs{x}^s - \bs{x}^{\square} ) + \alpha_s \varepsilon^s + \sum_{\ell = 1 }^L \beta_{\ell} \bs{p}^{\square} \cdot \bs{e}_{\ell} = u^{\sq} + \sum_{s \in \mathcal{T} \backslash \{\sq\} } \alpha_s \varepsilon^s
		\end{equation*}
		Adding up the inequalities in \eqref{eq:lp-ineqs-direct} (with $s = \diamond$) we have
		\begin{equation*}
			u^{\di} \leq \sum_{t \in \mathcal{T} \backslash \{\square\}} \alpha_t u^{t} + \alpha_t \bs{p}^t \cdot ( \bs{x}^t - \bs{x}^{\di} ) 
		\end{equation*}
		Combining the previous two displayed equations
		\begin{equation*}
			u^{\di} - u^{\sq} \leq \sum_{t \in \mathcal{T} \backslash \{\sq\}} \alpha_t \left( \bs{p}^t \cdot ( \bs{x}^{t} - \bs{x}^{\di} ) - \varepsilon^t \right)
		\end{equation*}
		and so the solution value of the LP is bounded above.
	\end{proof}
	
	\begin{step} \label{step:lp-big}
		$\hat{U}_{\ua}(\mbx^{\di}, \mbx^{\sq}) \geq U_{\ua}(\mbx^{\di}, \mbx^{\sq}; \mathcal{U}_{st} )$.
	\end{step}
	\begin{proof}
		From Step \ref{step:lp-bounded} it is clear that the result holds when the condition in \eqref{eq:lp-bounded} fails. So, let us suppose that \eqref{eq:lp-bounded} holds. Let $U: X \to \ereal$ be a standard utility function which achieves the infimum in the definition of $Q$ and which satisfies $U( \bs{x}^{\diamond} )  - U( \bs{x}^{\square} ) > -\infty$. Clearly, this implies that $U(\bs{x}^{t}) \in \real$ for all $t \mathcal{T} \backslash \{\sq\}$ and because \eqref{eq:lp-bounded} holds and $U$ is concave and weakly increasing we also see that $U( \bs{x}^{\sq} ) \in \real$. 
		
		For each $t \in \mathcal{T}$ let $u^t = U(\bs{x}^t)$. For each $t \in \{1,2,\ldots, T\}$ let 
		\begin{equation*}
			\varepsilon^t = \sup_{ \mbx \in X } \Big\{ U( \mbx ) - \mbp^t \cdot \mbx \Big\} - \Big( U( \mbx^t ) - \mbp^t \cdot \mbx^t \Big)
		\end{equation*}
		and let $\varepsilon^{\diamond} = \varepsilon^{\square} = 0$. As $U$ achieves the infimum in the definition of $Q$ we see that $\sum_{t=1}^T \varepsilon^t = \text{TMP}$ holds. 
		
		Next, note that for any $t \in \{1,2,\ldots, T\}$ and $s \in \mathcal{T}$ we have 
		\begin{equation*}
			u^t + \bs{p}^t \cdot ( \bs{x}^s - \bs{x}^t ) + \varepsilon^t = \bs{p}^t \cdot \bs{x}^s + \sup_{ \mbx \in X } \Big\{ U( \mbx ) - \mbp^t \cdot \mbx \Big\} \geq \bs{p}^t \cdot \bs{x}^s + u^s - \bs{p}^t \cdot \bs{x}^s = u^s
		\end{equation*}
		and so
		\begin{equation} \label{eq:key-ineqs}
			u^s \leq u^t + \bs{p}^t \cdot ( \bs{x}^s - \bs{x}^t ) + \varepsilon^t, \qquad \forall t \in \{1,2,\ldots, T\} \text{ and } s \in \mathcal{T}
		\end{equation}
		For each $t \in \mathcal{T} \backslash \{\square\}$ let $u^t = U(\bs{x}^t)$. From \eqref{eq:key-ineqs} we see that \eqref{eq:lp-ineqs-direct} holds for $s \in \mathcal{T} \backslash \{\square\}$ and $t \in \{1,2,\ldots, T\}$. As $U$ is concave and weakly increasing we may apply Step \ref{step:support-price} to find $\bs{p}^{\diamond}, \bs{p}^{\square} \in \eucp$ so that \eqref{eq:lp-ineqs-direct} holds for all $s,t \in \mathcal{T}$. We have thus found a feasible solution to the LP whose  solution value is equal to $U( \bs{x}^{\diamond} ) - U(\bs{x}^{\square})$ and because $U$ was an arbitrary element of $\mathcal{U}_{st}$ which satisfies $U(\bs{x}^{\di}) - U(\bs{x}^{\sq}) > -\infty$ the result follows. 
	\end{proof}
	
	\begin{step} \label{step:attain}
		There exists a standard $U: X \to \ereal$ which achieves the infimum in the definition of $Q$ and which satisfies $U( \bs{x}^{\di} ) - U( \bs{x}^{\sq} ) = U_{\ua}( \bs{x}^{\di}, \bs{x}^{\sq} )$.
	\end{step}
	\begin{proof}
		Suppose \eqref{eq:lp-bounded} does not hold. Let $\ti{U}: X \to \real$ be a well-behaved utility function which achieves the infimum in the definition of $Q$. Define $U$ by 
		\begin{equation*}
			U( \bs{x} ) = \begin{cases}
				\ti{U}(\bs{x}), \qquad & \text{ if } \bs{x} \in \conv(\bs{x}^1, \ldots, \bs{x}^T, \bs{x}^{\di}) + \eucp \\
				-\infty, \qquad & \text{ else.}
			\end{cases}
		\end{equation*}
		which satisfies the desired properties.
		
		Next, suppose \eqref{eq:lp-bounded} holds. From Step \ref{step:lp-bounded} we see that the LP has an optimal solution $( (u^t, \varepsilon^t)_{t \in \mathcal{T}}, \bs{p}^{\di}, \bs{p}^{\sq} )$. Step \ref{step:U-gen} now provides the desired $U$.
	\end{proof}
	We can now complete the proof. Obviously $U_{\ua}(\bs{x}^{\di}, \bs{x}^{\sq}, \mathcal{U}_{st}) \geq U_{\ua}( \bs{x}^{\di}, \bs{x}^{\sq}, \mathcal{U}_{wb} )$ and so the displayed equation in the theorem follows from Steps \ref{step:wb-ub} and \ref{step:lp-big}. Step \ref{step:attain} shows that the supremum in the definition of $U_{\ua}( \bs{x}^{\di}, \bs{x}^{\sq}; \mathcal{U}_{st} )$ is attained.
\end{proof}

\subsection{ Section \ref{sec:pop} Proofs }

We shall first present some key lemmas which we use in our proofs. We begin by defining the collection of integrable utility functions.

\subsubsection{Integrable Utility Functions}

For a probability measure $\nu$ on $\mathbb{R}^L$ we let $L^1(\nu)$ denote the collection of $\nu$-integrable utility functions $U: \mathbb{R}^L \to \ereal$. In other words, $U \in L^1(\nu)$ satisfies $\int |U(\bs{x})| \ d\nu(\bs{x}) < \infty$. Let $L^{st}(\nu)$ denote those $U \in L^1(\nu)$ which are standard. For a probability measure $\pi$ on $\mathbb{R}^L \times \mathbb{R}^L$ with marginals $\mu$ and $\nu$ we let $L^1(\pi)$ serve as a shorthand for $L^1(\nu)$. In other words, $U \in L^1(\pi)$ means that $U$ is measurable and that $\int |U(\bs{x})| \ d\pi( \bs{p}, \bs{x} ) < \infty$. Throughout, we identify each probability measure $\pi$ on $\eucpp \times X$ with the probability measure $\bar{\pi}$ on $\mathbb{R}^L \times \mathbb{R}^L$ which satisfies $\bar{\pi}( A ) = \pi( A \cap (\eucpp \times X) )$. Similar remarks hold for probability measures defined on $\eucpp$ and $X$.

\subsubsection{$c$-concavity}

Let $c: \eucpp \times X \to \real$ be some function. A function $U:X \to [-\infty, \infty)$ is \emph{$c$-concave} if there is a function $
\zeta: \eucpp \to [-\infty, \infty]$ so that
\begin{equation} \label{eq:c-concave}
	U(\bs{x}) = \inf_{ \bs{p} \in \eucpp } \Big( \zeta( \bs{p} ) + c( \bs{p}, \bs{x} ) \Big), \qquad\qquad \forall \bs{x} \in X
\end{equation}
The following shows that a $c$-concave function is standard for a certain function $c$. 
\begin{lemma} \label{lemma:c-concavity}
	Let $c: \eucpp \times X \to \real$ be defined by $c(\bs{p}, \bs{x}) = \bs{p} \cdot \bs{x}$. If $U: X \to \ereal$ is $c$-concave then $U$ is standard.
\end{lemma}
\begin{proof}
	Define $U: \conv(X) \to \ereal$ by \eqref{eq:c-concave}. By definition $U$ is the pointwise infimum of a collection of concave, continuous, and weakly increasing functions and thus $U$ is concave, upper semicontinuous, and weakly increasing. Thus, when $U$ is restricted to $X$ it is standard.
\end{proof}
f

\subsubsection{Optimal Transport Results}

\paragraph{Kantorovich Duality.}

\begin{lemma}[Kantorovich Duality] \label{lemma:kantorovich}
	Let $\bar{\pi}$ be a probability measure on $\eucpp \times X$ with finite variance and with marginals $\mu$ and $\nu$. Then
	\begin{equation} \label{eq:kantorovich}
		\max_{\pi \in \Pi(\mu, \nu)}  \E_{\bar{\pi}}[ \bs{p} \cdot \bs{x} ] - \E_{\pi}[ \bs{p} \cdot \bs{x} ]  = \inf_{ U \in L^1(\nu) } \E_{ \bar{\pi} }\Big[ \sup_{ \tbs{x} \in X } \big[ U(\tbs{x}) - \bs{p} \cdot \tbs{x} \big] - \big[ U(\bs{x}) - \bs{p} \cdot \bs{x} \big] \Big]
	\end{equation}
	The maximum in \eqref{eq:kantorovich} is attained by some $\pi^* \in \Pi(\mu,\nu)$ and further a $\pi \in \Pi(\mu,\nu)$ achieves this maximum if and only if $\pi$ satisfies cyclical monotonicity. Also, the minimum in \eqref{eq:kantorovich} is achieved by a standard $U: X \to [-\infty, \infty)$. 
\end{lemma}
\begin{proof}
	This is a version of Kantorovich duality. A more standard statement can be obtained by subtracting the constant $\max_{\pi \in \Pi(\mu, \nu)} \E_{\bar{\pi}}[ \bs{p} \cdot \bs{x} ]$ from both sides of \eqref{eq:kantorovich} and using the fact that $\bar{\pi} \in \Pi(\mu,\nu)$ to obtain
	\begin{equation*}
		\max_{\pi \in \Pi(\mu, \nu)} - \E_{\pi}[ \bs{p} \cdot \bs{x} ] = \min_{U \in L^1(\nu)} \E_{ \mu }\left[ \sup_{ \tbs{x} \in X } \big\{ U(\tbs{x}) - \bs{p} \cdot \tbs{x} \big\} \right] - \E_{\nu} \big[ U(\bs{x}) \big]
	\end{equation*}
	This is essentially the form stated in Theorem 5.10 of \citet{villani09} (take $c(\bs{p}, \bs{x}) = \bs{p} \cdot \bs{x}$ in Villani's notation). In Villani's result the maximum appears on the right and the minimum on the left. This is easily ``reversed'' by multiplying both sides by $-1$. Villani's Theorem shows that a $c$-concave $U$ achieves the infimum in \eqref{eq:kantorovich} and thus Lemma \ref{lemma:c-concavity} shows that this infimum is achieved by a standard $U$. Villani's Theorem 5.10 is actually quite lengthy and more general than the result we state here. The parts of Villani's Theorem which we use are (iia), (iib), and all of part (iii). 
\end{proof}

\paragraph{Uniqueness and convergence results.}

Let $\mathcal{P}_2$ denote the Borel probability measures on $\mathbb{R}^L$ with finite second moments (i.e.\ $\int || \mbx ||^2 \ d\mu < \infty$ for $\mu \in \mathcal{P}_2$). Similarly, $\mathcal{P}_2^2$ is the collection of Borel probability measures on $\mathbb{R}^L \times \mathbb{R}^L$ with finite second moments.

\begin{definition} \label{def:optimal}
	Let $\mu, \nu \in \mathcal{P}_2$. A coupling $\pi \in \Pi(\mu, \nu)$ is an \emph{optimal coupling of $(\mu,\nu)$} if it achieves the maximum in \eqref{eq:kantorovich}. A concave and upper semicontinuous function $U$ is a \emph{potential} for $(\mu,\nu)$ if it achieves the minimum in \eqref{eq:kantorovich}.
\end{definition}

\begin{lemma} \label{lemma:grad-ids}
	Let $A \subseteq \mathbb{R}^L$ be a non-empty open convex set. Suppose $U$ and $U'$ are convex functions mapping $\mathbb{R}^L$ to $\ereal$ which are real-valued on $A$ and satisfy
	\begin{equation*}
		\nabla U(\bs{x}) = \nabla U(\bs{x}), \qquad \text{ for almost every } \bs{x} \in A
	\end{equation*}
	then there exists $b \in \real$ so that $U(\bs{x}) = U'(\bs{x}) + b$ for all $\bs{x} \in A$.
\end{lemma}
This is Lemma 2.1 in \citet{barrio19}. We prove that the conclusion continues to hold if $A$ is connected but not necessarily convex.
\begin{lemma} \label{lemma:grad-ids2}
	Let $A \subseteq \mathbb{R}^L$ be a non-empty connected open set. Suppose $U$ and $U'$ are convex functions mapping $\mathbb{R}^L$ to $\ereal$ which are real-valued on $A$ and satisfy
	\begin{equation*}
		\nabla U(\bs{x}) = \nabla U(\bs{x}), \qquad \text{ for almost every } \bs{x} \in A
	\end{equation*}
	then there exists $b \in \real$ so that $U(\bs{x}) = U'(\bs{x}) + b$ for all $\bs{x} \in A$.
\end{lemma}
\begin{proof}
	Let $\bbs{x} \in A$ and let $b = U(\bbs{x}) - U'(\bbs{x})$. Let $C = \{ \bs{x} \in A: U(\bs{x}) = U'(\bs{x}) + b \}$. Note that $\bbs{x} \in C$ and thus $C \neq \varnothing$. Let $\tbs{x} \in \bar{C} \cap A$ (where $\bar{C}$ is the closure of $C$). As $A$ is open we can find some open ball $B$ so that $\tbs{x} \in B \subseteq A$. As $\tbs{x} \in \bar{C}$ there exists $\tbs{x}' \in B \cap C$. As $B$ is convex we can apply Lemma \ref{lemma:grad-ids} to see that $U( \bs{x} ) = U'( \bs{x} ) + \ti{b}$ for some $\ti{b} \in \real$ and all $\bs{x} \in B$. Now, as $\tbs{x}' \in B \cap C$ we have $U(\tbs{x}') = U(\tbs{x}') + b$ and therefore we see that $b = \ti{b}$ and we conclude that $U(\bs{x}) = U(\bs{x}) + b$ for all $\bs{x} \in B$. Thus, $B \subseteq C$. By supposing $\tbs{x} \in C$ we conclude that $C$ is open (as $B$ is then an open set containing $\tbs{x}$ which lies in $C$) and by supposing $\tbs{x} \in (\partial C) \cap A$ we conclude that $C$ is closed relative to $A$ (as we have seen that $\tbs{x} \in C$). So, $C$ is both open and closed relative to $A$ and, as $A$ is connected and $C$ is non-empty we conclude that $C = A$ which completes the proof.
\end{proof}

\begin{lemma}[Uniqueness] \label{lemma:unique-U2}
	Let $\mu,\nu \in \mathcal{P}_2$ where $\nu$ has a positive density on the connected non-empty interior of its support and let $U, U'$ be two potentials for $(\mu, \nu)$. Then, there exists some $b \in \real$ so that
	\begin{equation*} 
		U(\bs{x}) = U'(\bs{x}) + b, \qquad\qquad \text{ for all } \bs{x} \in \supp(\nu)^o
	\end{equation*}
\end{lemma}
\begin{proof}
	Let $A = \supp(\nu)^o$, let $\pi$ be an optimal coupling of $(\mu,\nu)$, and let $\ti{U}$ be any potential for $(\mu,\nu)$. Note that $\ti{U}$ must be real-valued on $A$ and so Theorem 25.5 in \citet{rockafellar70} shows that $\ti{U}$ is continuously differentiable almost everywhere on $A$. That is, there is a set $A' \subseteq A$ where $A \backslash A'$ has Lebesgue measure $0$ and where $\nabla \ti{U}$ exists and is continuous on $A'$. From Proposition \ref{prop:argmin-char-pop} we see $\nabla \ti{U}(\bs{x}) = \bs{p}$ for all $(\bs{p}, \bs{x}) \in \supp( \pi )$ where $\bs{x} \in A'$. As $\nu$ has positive density on $A$ we see $\nu( A' ) = 1$ and therefore, for any measurable $C \subseteq A$ 
	\begin{IEEEeqnarray*}{rCl}
		\mu(C) & = & \pi( C \times X ) = \pi( C \times A' ) = \nu\{ \bs{x} \in A': \nabla \ti{U}(\bs{x}) \in C \} = \nu \circ (\nabla \ti{U})^{-1}(C)
	\end{IEEEeqnarray*}
	So, $\mu = \nu \circ (\nabla \ti{U})^{-1}$. As $\ti{U}$ was an arbitrary potential for $(\mu,\nu)$ we see that $\mu = \nu \circ (\nabla U)^{-1} = \nu \circ (\nabla U')^{-1}$. Now, from Brenier's Theorem (see Theorem 2.12 in \citet{villani03}) we see that $\nabla \ti{U}(\bs{x}) = \nabla U(\bs{x}) = \nabla U'(\bs{x})$ for almost all $\bs{x} \in A$ and so Lemma \ref{lemma:grad-ids2} shows that there exists a $b \in \real$ so that $U(\bs{x}) = U(\bs{x}') + b$ for all $\bs{x} \in A$.  
\end{proof}
Lemma \ref{lemma:unique-U2} is a slight generalization of Corollary 2.2 in \citet{barrio19}. Del Barrio and Loubes' corollary requires additionally that $\supp(\nu)$ is convex.

The following defines a notion of convergence for sequences in $\mathcal{P}_2$.
\begin{definition} \label{def:measure-converge}
	For a sequence $\nu_n \in \mathcal{P}_2$ we write $\nu_n \rightarrow \nu \in \mathcal{P}_2$ to mean that (i) $\nu_n$ converges to $\nu$ in the usual weak topology (i.e.\ $\int f d \nu_n \rightarrow \int f d \nu$ for all continuous and bounded functions $f: X \rightarrow \mathbb{R}$) and (ii) the second moments of $\nu_n$ converge to the second moment of $\nu$ (i.e.\ $\int || \mbx ||^2 d \nu_n \rightarrow \int || \mbx ||^2 d \nu $). Convergence in $\mathcal{P}_2^2$ is similarly defined.
\end{definition}

\begin{lemma} \label{lemma:plan-conv}
	Let $\pi_n \in \mathcal{P}_2^2$ be a sequence of probability measures with $\pi_n \to \bar{\pi} \in \mathcal{P}_2^2$. Let $\mu_n, \nu_n$ be the marginals of $\pi_n$ and let $\mu,\nu$ be the marginals of $\bar{\pi}$. Then,
	\begin{equation} \label{eq:plan-conv}
		\max_{ \pi \in \Pi(\mu_n, \nu_n) } \Big( \E_{ \pi_n }[ \bs{p} \cdot \bs{x} ] - \E_{ \pi }[ \bs{p} \cdot \bs{x} ] \Big) \longrightarrow \max_{ \pi \in \Pi(\mu, \nu) } \Big( \E_{ \bar{\pi} }[ \bs{p} \cdot \bs{x} ] - \E_{ \pi }[ \bs{p} \cdot \bs{x} ] \Big)
	\end{equation}
\end{lemma}
\begin{proof}
	Let $\mathcal{W}_2(\ti{\mu}, \ti{\nu})$ be the Wasserstein distance of order 2 between $\ti{\mu}$ and $\ti{\nu}$. That is,
	\begin{equation*}
		\mathcal{W}_2(\ti{\mu}, \ti{\nu}) = \left( \inf_{ \pi \in \Pi( \ti{\mu}, \ti{\nu}) } \int || \bs{x} - \bs{y} ||^2 \ d \pi(\bs{y}, \bs{x}) \right)^{1/2}
	\end{equation*}
	For each $n$, let $\ti{\mu}_n(A) = \mu_n(-A)$ and $\ti{\mu}(A) = \mu(-A)$ where $-A = \{-\bs{y}: \bs{y} \in A\}$. Corollary 6.11 in \citet{villani09} shows
	\begin{equation} \label{eq:W-conv}
		\mathcal{W}_2( \ti{\mu}_n, \nu_n ) \to \mathcal{W}_2( \ti{\mu}, \nu)
	\end{equation}
	From the fact that $2 \bs{y} \cdot \bs{x} + ||\bs{x}||^2 + || \bs{y} ||^2 = || \bs{x} + \bs{y} ||^2$ we can use equation \eqref{eq:W-conv} to see
	\begin{equation*}
		\inf_{ \pi \in \Pi( \ti{\mu}_n, \nu_n ) } \int \bs{y} \cdot \bs{x} \ d \pi(\bs{y}, \bs{x}) \to \inf_{ \pi \in \Pi( \ti{\mu}, \nu ) } \int \bs{y} \cdot \bs{x} \ d \pi(\bs{y}, \bs{x})
	\end{equation*}
	and so
	\begin{IEEEeqnarray*}{rCl}
		\inf_{ \pi \in \Pi(\mu_n, \nu_n) } \int -\bs{p} \cdot \bs{x} \ d \pi(\bs{p}, \bs{x}) & = & \inf_{ \pi \in \Pi( \ti{\mu}_n, \nu_n ) } \int \bs{y} \cdot \bs{x} \ d \pi(\bs{y}, \bs{x}) \\
		& \to & \inf_{ \pi \in \Pi( \ti{\mu}, \nu ) } \int \bs{y} \cdot \bs{x} \ d \pi(\bs{y}, \bs{x}) = \inf_{ \pi \in \Pi( \mu, \nu ) } \int -\bs{p} \cdot \bs{x} \ d \pi(\bs{p}, \bs{x})
	\end{IEEEeqnarray*}
	which implies \eqref{eq:plan-conv}.
\end{proof}

\begin{lemma} \label{lemma:potential-conv2}
	Let $\mu_n, \nu_n \in \mathcal{P}_2$ be sequences satisfying $\mu_n \to \mu \in \mathcal{P}_2$ and $\nu_n \to \nu \in \mathcal{P}_2$ where $\nu$ has a positive density on the non-empty connected interior of its support. Let $\bbs{x} \in \supp( \nu )^o$. For each $n$ let $U_n$ be a potential for $(\mu_n, \nu_n)$ which satisfies $U_n(\bbs{x}) = 0$ and let $U$ be a potential for $(\mu, \nu)$ which satisfies $U(\bbs{x}) = 0$. Then, 
	\begin{equation*} 
		U_n(\bs{x}) \to U(\bs{x}), \qquad\qquad \text{ for all } \bs{x} \in \supp(\nu)^o
	\end{equation*}
\end{lemma}
\begin{proof}
	Theorem 2.8 of \citet{barrio19} shows that the conclusion of this lemma holds when $\nu$ has a positive density on the interior of its convex support. So, the present lemma is slightly more general in the sense that we require $\supp(\nu)^o$ to be connected instead of convex. Actually, the proof of Theorem 2.8 of \citet{barrio19} largely goes through without alternation when $\supp(\nu)^o$ is connected (and not necessarily convex). The only modification needed is to invoke our Lemma \ref{lemma:unique-U2} instead of del Barrio and Loubes' Lemma 2.1 near the end of the proof.
\end{proof}

\subsubsection{A Convergence of Measure Result}

\begin{lemma} \label{lemma:conv}
	Let $(\mbp^t, \mbx^t)_{t \in \mathbb{Z}}$ be an ergodic stationary process on probability space $( \Omega, \mathcal{F}, P )$ where $(\mbp^t, \mbx^t)$ has distribution $\pi_0 \in \mathcal{P}_2^2$. Let $\pi_T$ be the empirical measure for $( \mbp^1, \mbx^1 ), \ldots, (\mbp^T, \mbx^T)$ (that is, $\pi_T(A) = \tfrac{1}{T} \sum_{t=1}^T \mathbf{1}[ (\mbp^t, \mbx^t) \in A ]$ where $\mathbf{1}[\cdot]$ is the indicator function). Then, $P$-almost surely, $\pi_T \rightarrow \pi_0 \in \mathcal{P}_2^2$.
\end{lemma}
\begin{proof}
	Let $C_b$ denote the collection of continuous and bounded real-valued functions with domain $\mathbb{R}^L \times \mathbb{R}^L$. From Theorem 6.6 in \citet{parthasarathy67} there exists a sequence $\{g_n\}_{n \in \mathbb{N}}$ in $C_b$ which characterizes weak convergence in the sense that 
	\begin{equation} \label{eq:conv-char}
		\int g \ d \pi_T \to \int g \ d \pi_0, \ \forall g \in C_b \qquad \iff \qquad \int g_n \ d \pi_T \to \int g_n \ d \pi_0, \ \forall n \in \mathbb{N}
	\end{equation}
	For each $n \in \mathbb{N}$ we may apply the Ergodic Theorem (see Theorem 20.14 of \citet{klenke-probabilitytheory}) to see that, almost surely, 
	\begin{equation*}
		\int g_n \ d \pi_T = \dfrac{1}{T} \sum_{t=1}^T g_n( \bs{p}^t, \bs{x}^t ) \to \bs{E}\Big[ g_n(\bs{p}^t, \bs{x}^t) \Big] = \int g_n \ d \pi_0
	\end{equation*}
	and so the right hand side condition of \eqref{eq:conv-char} occurs with probability 1 for each $n$. Therefore, the left hand side condition of \eqref{eq:conv-char} occurs with probability 1. 
	
	To complete the proof we again apply the Ergodic Theorem to see that, almost surely,
	\begin{equation*}
		\int || (\bs{p}, \bs{x}) || \ d \pi_T = \dfrac{1}{T} \sum_{t=1}^T || (\bs{p}^t, \bs{x}^t) || \to \bs{E} \Big[ || (\bs{p}^t, \bs{x}^t) || \Big] = \int || (\bs{p}^t, \bs{x}^t) || \ d \pi_0
	\end{equation*}
	and thus, almost surely, $\pi_T \to \pi_0 \in \mathcal{P}_2^2$.
\end{proof}

Let $\pi_0$ be a probability measure on $\eucpp \times X$ with finite second moments and with marginals $\mu_0$ and $\nu_0$. For a $U \in L^1(\nu_0)$ let
\begin{equation} \label{eq:Q-bar}
	\bar{Q}(U) = \E_{\pi_0} \left[ \left( \sup_{ \tbs{x} \in X} \Big\{ U(\tbs{x}) - \bs{p} \cdot \tbs{x} \Big\} \right) - \Big( U(\bs{x}) - \bs{p} \cdot \bs{x} \Big) \right] 
\end{equation}
and 
\begin{equation} \label{eq:A-bar}
	\bar{A}(U) = \bs{E}_{\pi_0}[ \bs{p} \cdot \bs{x} ] - e_U
\end{equation}
where $e_U$ is as defined in Section \ref{sec:tmp-pop}. Note that $Q_0 = \inf_{U \in L^1(\nu_0) } \bar{Q}(U)$ and $A_0 = \inf_{U \in L^1(\nu_0) } \bar{A}(U)$.

\begin{lemma} \label{lemma:Q-A-T}
	Suppose $\pi_0$ is a distribution on $\eucpp \times X$ with finite variance and with marginals $\mu_0$ and $\nu_0$. Then,  
	\begin{equation} \label{eq:Q-A-T}
		\bar{Q}(U) \geq \bar{A}(U) \geq \text{TMP}_0, \qquad \qquad \forall U \in L^1(\nu_0)
	\end{equation}
\end{lemma}
\begin{proof}
	Let $U \in L^1(\nu_0)$. For $\pi \in \Pi(\mu_0)$ which satisfies $\E_{\pi}[ U(\bs{x}) ] \geq \E_{\pi_0}[ U(\bs{x})]$ we have
	\begin{equation*}
		\bar{Q}(U) \geq \E_{\ti{\pi}} \Big[ U(\bs{x}) - \bs{p} \cdot \bs{x} \Big] - \E_{\pi_0} \Big[ U(\bs{x}) - \bs{p} \cdot \bs{x} \Big] \geq \E_{\pi_0} [ \bs{p} \cdot \bs{x} ] - \E_{\pi}[ \bs{p} \cdot \bs{x} ]
	\end{equation*}
	and so $\bar{Q}(U) \geq \bar{A}(U)$. For $U \in L^1(\nu_0)$ and $\pi \in \Pi(\mu_0,\nu_0)$ we have 
	\begin{equation*}
		\bar{A}(U) = \bs{E}_{\pi_0}[ \bs{p} \cdot \bs{x} ] - e_U \geq \E_{\pi_0}[ \bs{p} \cdot \bs{x} ] - \E_{\pi}[ \bs{p} \cdot \bs{x} ]
	\end{equation*}
	and so $\bar{A}(U) \geq \text{TMP}_0$. 
\end{proof}

\noindent \textbf{Theorem \ref{theorem:pop-tmp-equi}.}

\begin{proof}[Proof of Theorem \ref{theorem:pop-tmp-equi}.]
	Lemma \ref{lemma:kantorovich} guarantees that \eqref{eq:kantorovich} holds and that the maximum in \eqref{eq:kantorovich} is achieved by some $\pi^* \in \Pi(\mu_0, \nu_0)$ and the minimum is achieved by some standard $U^*: X \to \ereal$. Further, we have that $\pi$ achieves the maximum in \eqref{eq:kantorovich} if and only if $\pi$ satisfies cyclical monotonicity. Thus, we see that (i) $Q_0 = \text{TMP}_0$, (ii) the infimum in the definition of $Q_0$ is achieved by a standard $U^*$, (iii) the supremum in the definition of $\text{TMP}_0$ is achieved by a $\pi^*$ and further this supremum is achieved by a $\pi$ if and only if $\pi$ satisfies cyclical monotonicity. Further, from Lemma \ref{lemma:Q-A-T} we see that $U^*$ also achieves the infimum in the definition of $A_0$ and that $A_0 = \text{TMP}_0$.
	
	To prove the final part of the theorem we suppose that $\nu_0$ has a positive density on the interior of its convex support and we let $U$ and $U'$ be two standard utility functions which achieve the infimum in the definition of $Q_0$. Clearly, $U$ and $U'$ are potentials for $(\mu_0, \nu_0)$ and thus Lemma \ref{lemma:unique-U2} establishes that there is a number $b \in \real$ so that $U(\bs{x}) = U'(\bs{x}) + b$ for all $\bs{x} \in \supp(\nu_0)^o$. 
\end{proof}

\noindent \textbf{Proposition \ref{prop:reny}.}

\begin{proof}[Proof of Proposition \ref{prop:reny}.]
	It is obvious that 3 implies 2 implies 1 and so we just show that if $D$ satisfies cyclical monotonicity then $D$ can be quasilinear rationalized by a standard $U$. As $D$ is a non-empty subset of $\eucpp \times X$ we can find a countable set $Y = \{ \bs{y}_1, \bs{y}_2, \bs{y}_3, \ldots \} \subseteq D$ which is dense in $D$. Let $\pi$ be some probability measure on $\eucpp \times X$ defined so that $\pi( Y ) = 1$, $\pi(\bs{y}_k) > 0$ for all $k$, and so that $\pi$ has finite second moments. For instance, the probability measure defined by
	\begin{equation*}
		\pi( \bs{y}_k ) = \dfrac{ 2^{-k} \min(1, || \bs{y}_k ||^{-2} ) }{ \sum_{j=1}^{\infty} 2^{-j} \min(1, ||\bs{y}_j||^{-2}) } 
	\end{equation*}
	will do.\footnote{This uses the convention that $0^{-2} = \infty$.} Clearly, $\pi$ satisfies cyclical monotonicity and so Theorem \ref{theorem:pop-tmp-equi} establishes that there is a standard $U$ which quasilinear rationalizes $\pi$. Clearly, this $U$ quasilinear rationalizes $D$.
\end{proof}

\noindent \textbf{Theorem \ref{theorem:tmp-conv}.}

\begin{proof}[Proof of Theorem \ref{theorem:tmp-conv}.]
	Let $\pi_T$ denote the empirical measure for $(\bs{p}^1, \bs{x}^1), \ldots, (\bs{p}^T, \bs{x}^T)$ (that is, $\pi_T(A) = \tfrac{1}{T} \bs{1}[ (\bs{p}^t, \bs{x}^t) \in A ]$ where $\bs{1}[\cdot]$ is the indicator function). Similarly, let $\mu_T$ denote the empirical measure for $\bs{p}^1, \bs{p}^2, \ldots, \bs{p}^T$ and let $\nu_T$ denote the empirical measure for $\bs{x}^1, \bs{x}^2, \ldots, \bs{x}^T$ . 
	
	For some fixed $T$ let $U$ be a standard utility function which achieves the infimum in the definition of $Q$ for the dataset $D_T$. For any $\pi \in \Pi( \mu_T, \nu_T )$ we may use Theorem \ref{theorem:tmp} to see
	\begin{IEEEeqnarray*}{rCl}
		\dfrac{\text{TMP}_T}{T} & = & \dfrac{Q}{T} = \dfrac{1}{T} \left( \sum_{t=1}^T \sup_{\tbs{x} \in X} [ U(\tbs{x}) - \bs{p}^t \cdot \tbs{x} ] - [ U(\bs{x}^t) - \bs{p}^t \cdot \bs{x}^t ] \right) \\
		& = & \E_{\pi_T} \left[ \sup_{ \tbs{x} \in X } [ U(\tbs{x}) - \bs{p} \cdot \tbs{x} ] - [U(\bs{x}) - \bs{p} \cdot \bs{x}] \right] \\
		& \geq & \E_{\pi} [U(\bs{x}) - \bs{p} \cdot \bs{x}] - \E_{\pi_T} [U( \bs{x} ) - \bs{p} \cdot \bs{x}] \\
		& = & \E_{\pi_T} [\bs{p} \cdot \bs{x}] - \E_{\pi} [\bs{p} \cdot \bs{x}]
	\end{IEEEeqnarray*}
	From which it is clear that 
	\begin{equation*}
		\dfrac{\text{TMP}_T}{T} = \sup_{ \pi \in \Pi( \mu_T, \nu_T ) } \Big\{ \E_{\pi_T} [\bs{p} \cdot \bs{x}] - \E_{\pi} [\bs{p} \cdot \bs{x}] \Big\}
	\end{equation*}
	From Lemma \ref{lemma:conv} we see that, almost surely, $\pi_T \to \pi_0 \in \mathcal{P}_2^2$ and thus almost surely, $\mu_T \to \mu_0 \in \mathcal{P}_2$ and $\nu_T \to \nu_0 \in \mathcal{P}_2$. We may then appeal to Lemma \ref{lemma:plan-conv} to see that, almost surely,
	\begin{equation*} 
		\dfrac{\text{TMP}_T}{T} = \sup_{ \pi \in \Pi( \mu_T, \nu_T ) } \Big\{ \E_{\pi_T} [\bs{p} \cdot \bs{x}] - \E_{\pi} [\bs{p} \cdot \bs{x}] \Big\} \to \sup_{ \pi \in \Pi( \mu_0, \nu_0 ) } \Big\{ \E_{\pi_0} [\bs{p} \cdot \bs{x}] - \E_{\pi} [\bs{p} \cdot \bs{x}] \Big\}
	\end{equation*}
	and so almost surely, $\text{TMP}_T / T \to \text{TMP}_0$.
	
	Next, let $U_T$ and $U$ be as in the statement of the theorem. It is clear that $U_T$ is a potential for $(\mu_T, \nu_T)$ and that $U$ is a potential for $(\mu, \nu)$. Therefore Lemma \ref{lemma:potential-conv2} establishes that, almost surely, $U_T(\bs{x}) \to U(\bs{x})$ for all $\bs{x} \in \supp(\nu_0)^o$.
\end{proof}

\noindent \textbf{Proposition \ref{prop:unique-U-finite}.}

\begin{proof}[Proof of Proposition \ref{prop:unique-U-finite}.]
	Suppose that $U \in L^{st}(\nu_0)$ achieves the infimum in the definition of $Q_0$ and suppose that $U' \in L^{st}(\nu_0)$ does not satisfy \eqref{eq:U-id} for any $b \in \real$. We shall show that $U'$ does not achieve the infimum in the definition of $Q_0$ which will prove the first part of the proposition. Let $A \subseteq \supp(\nu_0)$ be the collection of bundles $\bs{x} \in \supp(\nu_0)$ which satisfy
	\begin{equation*}
		U'(\bs{x}) - U(\bs{x}) = \min_{ \tbs{x} \in \supp(\nu_0) } U'(\tbs{x}) - U(\tbs{x})
	\end{equation*}
	and let $\ti{A} = \supp(\nu_0) \backslash A$. Note that $\ti{A}$ is not empty because $U'$ does not satisfy \eqref{eq:U-id} for any $b \in \real$. Let $\pi \in \Pi(\mu_0, \nu_0)$ attain the supremum in the definition of $\text{TMP}_0$. We claim that
	\begin{equation} \label{eq:boundary-price}
		\exists \bbs{x} \in A, \quad \tbs{x} \in \ti{A}, \quad \bbs{p} \in \supp(\mu_0) \quad \text{ so that } \quad (\bbs{p}, \bbs{x}), (\bbs{p}, \tbs{x}) \in \supp(\pi) 
	\end{equation}
	To show that \eqref{eq:boundary-price} is true let $P_A$ and $P_{\ti{A}}$ be defined by
	\begin{IEEEeqnarray*}{rCl}
		P_{A} & = & \Big\{ \bs{p} \in \supp(\mu_0): (\bs{p}, \bs{x}) \in \supp(\pi) \text{ for some } \bs{x} \in A \Big\} \\
		P_{ \ti{A} } & = & \Big\{ \bs{p} \in \supp(\mu_0): (\bs{p}, \bs{x}) \in \supp(\pi) \text{ for some } \bs{x} \in \ti{A} \Big\}
	\end{IEEEeqnarray*}
	Note that $P_A$ and $P_{\ti{A}}$ are both closed sets and that $\supp(\mu_0) = P_A \cup P_{\ti{A}}$. Thus, as $\supp(\mu_0)$ is connected, there must be some $\bbs{p} \in P_A \cap P_{\ti{A}}$ and therefore \eqref{eq:boundary-price} holds. From Proposition \ref{prop:argmin-char-pop} we see that $U$ quasilinear rationalizes $\pi$ and so
	\begin{equation*}
		U( \tbs{x} ) - \bbs{p} \cdot \tbs{x} = U( \bbs{x} ) - \bbs{p} \cdot \bbs{x}
	\end{equation*}
	Now, as $\bbs{x} \in A$ and $\tbs{x} \in \ti{A}$ we see that
	\begin{equation*}
		U'( \tbs{x} ) - \bbs{p} \cdot \tbs{x} > U'( \bbs{x} ) - \bbs{p} \cdot \bbs{x}
	\end{equation*}
	and so $U'$ does not quasilinear rationalize $\pi$. Proposition \ref{prop:argmin-char-pop} now shows that $U'$ does not achieve the argmin in the definition of $Q_0$.
	
	We next proceed to prove \eqref{eq:U-conv2}. Suppose the support of $\nu_0$ has $K+1$ elements which we enumerate as $\supp(\nu_0) = \{ \bs{x}_0, \bs{x}_1, \bs{x}_2, \ldots, \bs{x}_K \}$ and suppose that $U( \bs{x}_0 ) = U_T(\bs{x}_0) = 0$ for all $T$. We identify each $\bar{U} \in L^{st}(\nu_0)$ which satisfies $\bar{U}(\bs{x}_0) = 0$ with the vector in $\mathbb{R}^K$ whose entry $k$ is $\bar{U}(\bs{x}_k)$. So, we shall write $\bar{U} = ( \bar{U}_1, \bar{U}_2, \ldots, \bar{U}_K ) \in \mathbb{R}^K$ where $\bar{U}_k = \bar{U}(\bs{x}_k)$. Let $Q_T(\bar{U})$ be defined by 
	\begin{equation*}
		Q_T(\bar{U}) = \dfrac{1}{T} \sum_{t=1}^T \bigg( \sup_{ k \in \{0,1,2\ldots, K\} } \Big[ \bar{U}_k - \bs{p}^t \cdot \bs{x}_k \Big] - \Big[ \bar{U}(\bs{x}^t) - \bs{p}^t \cdot \bs{x}^t \Big] \bigg)
	\end{equation*}
	where $\bar{U}_0 = 0$ by definition and let $Q_0(\bar{U})$ be defined by
	\begin{equation*}
		Q_0(\bar{U}) = \E\left[ \sup_{ k \in \{0,1,2\ldots, K\} } \Big[ \bar{U}_k - \bs{p}^t \cdot \bs{x}_k \Big] - \Big[ \bar{U}(\bs{x}^t) - \bs{p}^t \cdot \bs{x}^t \Big] \right]
	\end{equation*}
	where again $\bar{U}_0 = 0$ by definition. Because (i) the pointwise supremum of convex functions is a convex function and (ii) the sum of convex functions is a convex function we see that $Q_T: \mathbb{R}^K \to \mathbb{R}$ is convex. As \eqref{eq:U-id} holds we see that 
	\begin{equation} \label{eq:U-id-loss}
		Q_0(U) < Q_0(\bar{U}), \qquad\qquad \text{ for all } \bar{U} \in \mathbb{R}^K \text{ s.t. } \bar{U}\neq U
	\end{equation}
	
	Let $\mathcal{U}$ be a countable dense subset of $\mathbb{R}^K$ and let $C \subseteq \mathbb{R}^K$ be a closed ball of positive radius centered at $U$. From the Ergodic Theorem we see that, almost surely, $Q_T(\bar{U}) \to Q_0(\bar{U})$ for all $\bar{U} \in \mathcal{U}$. Now, as $Q_T$ is convex for each $T$ we may appeal to Theorem 10.8 in \citet{rockafellar70} to see that, almost surely, $Q_T$ converges uniformly to $Q_0$ on $C$. From hereon we fix a point in the sample space on which $Q_T$ converges uniformly to $Q_0$ on $C$. As this point in the sample space is an arbitrary element of a probability 1 event we need only establish that $U_T \to U$ for this point in the sample space for our proof to be complete. 
	
	Now, let $\ti{C} \subseteq C$ be some closed ball of positive radius centered at $U$. As $U$ satisfies \eqref{eq:U-id-loss} we see that there exists $\bar{T}$ so that
	\begin{equation} \label{eq:U-id-loss2}
		Q_T(U) < \min_{ \bar{U} \in \partial \ti{C} } Q_T(\bar{U}), \qquad\qquad \text{ for all } T \geq \bar{T}
	\end{equation}
	where $\partial \ti{C}$ is the boundary of $\ti{C}$. Let $\bar{U} \in \mathbb{R}^K$ satisfy $\bar{U} \notin \ti{C}$ and let $\ti{U} \in \partial \ti{C}$ be the point ``between'' $U$ and $\bar{U}$ in the sense that it can be represented as $\ti{U} = \alpha U + (1-\alpha) \bar{U}$ for some $\alpha \in (0,1)$. From \eqref{eq:U-id-loss2} we see that, for all $T \geq \bar{T}$,
	\begin{equation*}
		Q_T(U) < Q_T( \ti{U} ) \leq \alpha Q_T( U ) + (1-\alpha) Q_T( \bar{U} )
	\end{equation*}
	where we used the fact that $Q_T$ is convex. Rearranging we see that for all $T \geq \bar{T}$ we have $Q_T(U) < Q_T(\bar{U})$. As $\bar{U}$ is an arbitrary element in $\mathbb{R}^K \backslash \ti{C}$ we conclude that $U_T$ belongs to $\ti{C}$ for all $T \geq \bar{T}$. But, because $\ti{C} \subseteq C$ is an arbitrary closed ball of positive radius centered at $U$ we see that $U_T \to U$.
\end{proof}

\noindent \textbf{Proposition \ref{prop:price-misperception}.}

\begin{lemma} \label{lemma:pop-tmp-price-misp}
	Suppose $(\tbs{p}^t, \bs{p}^t, \bs{x}^t)_{t \in \ints}$ and $U_0$ satisfy Assumption \ref{assumption:price-misperception} and assume $\tbs{p}^t$ and $\bs{p}^t$ follow the same distribution. Let $\pi^*$ denote the joint distribution of $(\tbs{p}^t, \bs{x}^t)$ and let $\bar{Q}$ be defined by \eqref{eq:Q-bar}. Then, $\text{MP}_0( \pi^* ) = \bar{Q}(U_0)$ and 
	\begin{equation} \label{eq:pop-tmp-price-misp}
		\text{TMP}_0 = \bs{E} \Big[ ( \bs{p}^t - \tbs{p}^t ) \cdot \bs{x}^t \Big]
	\end{equation}
\end{lemma}
\begin{proof}
	From \eqref{eq:price-misperception} we see
	\begin{equation*}
		\E_{\pi^*}[ U_0(\bs{x}) - \bs{p} \cdot \bs{x} ] = \E_{\pi_0} \left[ \max_{ \tbs{x} \in X } \Big\{ U_0(\tbs{x}) - \bs{p} \cdot \tbs{x} \Big\} \right]
	\end{equation*}
	and so
	\begin{equation*}
		\bar{Q}(U_0) = \E_{\pi^*}[ U_0(\bs{x}) - \bs{p} \cdot \bs{x} ] - \E_{\pi_0}[ U_0( \bs{x}) - \bs{p} \cdot \bs{x} ] = \E_{\pi_0}[ \bs{p} \cdot \bs{x} ] - \E_{\pi^*}[ \bs{p} \cdot \bs{x} ] = \text{MP}_0(\pi^*)
	\end{equation*}
	Using Lemma \ref{lemma:Q-A-T} we see that 
	\begin{equation*}
		\text{TMP}_0 = \text{MP}(\pi^*) = \E_{\pi_0}[ \bs{p} \cdot \bs{x} ] - \E_{\pi^*}[ \bs{p} \cdot \bs{x} ] = \E[ \bs{p}^t\cdot \bs{x} ] - \E[ \tbs{p}^t \cdot \bs{x}^t ] = \bs{E} \Big[ ( \bs{p}^t - \tbs{p}^t ) \cdot \bs{x}^t \Big]
	\end{equation*}
	which completes the proof.
\end{proof}

\begin{proof}[Proof of Proposition \ref{prop:price-misperception}.]
	Let $\pi^*$ denote the joint distribution of $( \tbs{p}^t, \bs{x}^t )$. Lemma \ref{lemma:pop-tmp-price-misp} establishes that $\text{MP}_0(\pi^*) = \bar{Q}(U_0)$ and Lemma \ref{lemma:Q-A-T} establishes that $U_0$ achieves the infimum in the definition of $Q_0$.
\end{proof}

\noindent \textbf{Proposition \ref{prop:lagged-prices}.}

Several of the propositions will assume that the utility function $U_0: \eucp \to [-\infty, \infty)$ is \emph{strictly concave}. Under this assumption the utility function $U_0$ generates a unique demand function $\bs{f}: \supp(\bs{p}^t) \to \eucp$ in the sense that there is one and only one function $\bs{f}: \supp(\bs{p}^t) \to \eucp$ which satisfies
\begin{equation} \label{eq:demand}
	U( \bs{f}(\bs{p}) ) - \bs{p} \cdot \bs{f}(\bs{p}) = \max_{ \bs{x} } \Big\{ U( \bs{x} ) - \bs{p} \cdot \bs{x} \Big\}
\end{equation}
In the following proofs we will often make use of the existence of this $\bs{f}$. 
\begin{lemma} \label{lemma:S_0}
	Suppose $(\tbs{p}^t, \bs{p}^t, \bs{x}^t)_{t \in \ints}$ and $U_0$ satisfy Assumption \ref{assumption:price-misperception} and that prices satisfy the AR(1) process of Assumption \ref{assumption:ar1}. Let $S_0$ be defined by \eqref{eq:S}. Then, 
	\begin{equation} \label{eq:S_0}
		S_0 = \E \Big[ ( \bbs{p} - \tbs{p}^t ) \cdot \bs{x}^t \Big]
	\end{equation}
\end{lemma}
\begin{proof}
	This is an obvious corollary of equation \eqref{eq:pop-tmp-price-misp} in Lemma \ref{lemma:pop-tmp-price-misp}. 
\end{proof}

\begin{lemma} \label{lemma:ar}
	Suppose prices $\{ \bs{p}^t \}_{t \in \mathbb{Z}}$ follow an AR(1) process
	\begin{equation*} \label{eq:ar1-2}
		\bs{p}^{t} = (1-\alpha) \bbs{p} + \alpha \bs{p}^{t-1} + \bs{\varepsilon}^t 
	\end{equation*}
	where $\alpha \in [0,1)$ and $\bbs{p}$ is a constant vector. Then, for any $t \in \mathbb{Z}$ and $k \in \{0,1,2,\ldots\}$ 
	\begin{equation} \label{eq:ar1-result}
		\bs{p}^t = (1-\alpha^k) \bbs{p} + \alpha^{k} \bs{p}^{t-k} + \sum_{s=0}^{k-1} \alpha^s \bs{\varepsilon}^{t-s}
	\end{equation}
\end{lemma}
\begin{proof}
	As $\bs{p}^t$ follows \eqref{eq:ar1-2} we see that $\bs{p}^{t-1} = (1-\alpha) \bbs{p} + \alpha \bs{p}^{t-2} + \bs{\varepsilon}^{t-1}$. Plugging this into \eqref{eq:ar1-2} we get $\bs{p}^t = (1-\alpha) (1 + \alpha) \bbs{p} + \alpha^2 \bs{p}^{t-2} + \bs{\varepsilon}^t + \alpha \bs{\varepsilon}^{t-1}$. Repeating this process we obtain
	\begin{IEEEeqnarray*}{rCl}
		\bs{p}^t & = & (1-\alpha) \left( \sum_{s=0}^{k-1} \alpha^{s} \right) \bbs{p} + \alpha^{k} \bs{p}^{t-k} + \sum_{s=0}^{k-1} \alpha^s \varepsilon^{t-s} \\
		& = & (1-\alpha^k) \bbs{p} + \alpha^{k} \bs{p}^{t-k} + \sum_{s=0}^{k-1} \alpha^s \bs{\varepsilon}^{t-s}
	\end{IEEEeqnarray*}
	where the second equality used the fact that $(1-\alpha) \left(\sum_{s=0}^{k-1} \alpha^s \right) = 1-\alpha^k$. 
\end{proof}

\begin{lemma} \label{lemma:optimal-demand}
	Suppose $(\tbs{p}^t, \bs{p}^t, \bs{x}^t)_{t \in \ints}$ and $U_0$ satisfy Assumption \ref{assumption:price-misperception}. Let $\ti{\mu}_0$, $\mu_0$, and $\nu_0$ denote the distributions of $\tbs{p}^t$, $\bs{p}^t$, and $\bs{x}^t$, respectively. Assume $\bs{p}^t$ and $\tbs{p}^t$ follow the same distribution (i.e.\ $\mu_0 = \ti{\mu}_0$). Then,
	\begin{equation} \label{eq:optimal-demand}
		\E\Big[ \tbs{p}^t \cdot \bs{x}^t \Big] = \max_{ \pi \in \Pi(\mu_0,\nu_0) } \E_{\pi} \Big[ \bs{p} \cdot \bs{x} \Big]
	\end{equation}
\end{lemma}
\begin{proof}
	This is an obvious corollary of equation \eqref{eq:pop-tmp-price-misp} in Lemma \ref{lemma:pop-tmp-price-misp}. 
\end{proof}

\begin{proof}[Proof of Proposition \ref{prop:lagged-prices}.]
	It is obvious that $\bs{p}^t \sim \tbs{p}^t$ and so we proceed to prove the remaining parts of the proposition. Let $\bs{f}$ be the demand function which satisfies \eqref{eq:demand}. Using \eqref{eq:pop-tmp-price-misp} from Lemma \ref{lemma:pop-tmp-price-misp}, Lemma \ref{lemma:ar}, and the fact that $\bs{\varepsilon}^{t-s} \indep \bs{p}^{t-k}$ for all $s \in \{0,1,\ldots, k-1\}$ we see
	\begin{IEEEeqnarray*}{rCl}
		\text{TMP}_0 & = & \E\Big[ (\bs{p}^t - \tbs{p}^t) \cdot \bs{x}^t \Big] =  \\
		& = & \E\left[ (1-\alpha^k) \bbs{p} \cdot \bs{x}^t + \alpha^k \bs{p}^{t-k} \cdot \bs{x}^t + \sum_{s=0}^{k-1} \alpha^s \bs{\varepsilon}^{t-s} \cdot \bs{x}^t - \tbs{p}^t \cdot \bs{x}^t \right] \\
		& = & \E\left[ (1-\alpha^k) (\bbs{p} - \bs{p}^{t-k} ) \cdot \bs{x}^t + \sum_{s=0}^{k-1} \alpha^s \bs{\varepsilon}^{t-s} \cdot \bs{f}(\bs{p}^{t-k}) \right] \\
		& = & \E\Big[ (1-\alpha^k) ( \bbs{p} - \bs{p}^{t-k} ) \cdot \bs{x}^t \Big] = (1- \alpha^k) S_0
	\end{IEEEeqnarray*}
	where the final equality used Lemma \ref{lemma:S_0}. 
\end{proof}

\noindent \textbf{Proposition \ref{prop:inattention}.}

\begin{proof}[Proof of Proposition \ref{prop:inattention}.]
	As $\tbs{p}^{t}$ is an independent mixture of $\bs{p}^t, \bs{p}^{t-1}$, $\bs{p}^{t-2}$, $\ldots$ and each of these variables have the same distribution (due to our stationarity assumption) we see that $\tbs{p}^t$ also has this distribution and so $\tbs{p}^t \sim \bs{p}^t$. 
	
	Let $E_k^t$ be the event that mental prices were last updated $k$ periods ago (i.e.\ $a_{t-s} = 0$ for $s \in \{0,1,2,\ldots, k-1\}$ and $a_{t-k} = 1$). Let $\bs{f}$ be the demand function which satisfies \eqref{eq:demand}. Using \eqref{eq:pop-tmp-price-misp} from Lemma \ref{lemma:pop-tmp-price-misp}, Lemma \ref{lemma:ar}, the fact that $\bs{\varepsilon}^{t-s} \indep \bs{p}^{t-k}$ for all $s \in \{0,1,\ldots, k-1\}$, the fact that $\{ a_{j} \}_{j \in \ints} \indep \bs{p}^{t}$ for all $t$, and \eqref{eq:S_0} we see 
	\begin{IEEEeqnarray*}{rCl}
		\text{TMP}_0 & = & \bs{E}\Big[ (\bs{p}^t - \tbs{p}^t) \cdot \bs{x}^t \Big] = \sum_{k=0}^{\infty} P( E_k^t ) \bs{E} \Big[ (\bs{p}^t - \tbs{p}^t) \cdot \bs{x}^t \Big| E_k^t \Big] \\
		& = & \sum_{k=0}^{\infty} (1-\beta) \beta^k \bs{E} \Big[ (\bs{p}^t - \bs{p}^{t-k}) \cdot \bs{f}(\bs{p}^{t-k}) \Big| E_k^t \Big] \\
		& = & \sum_{k=0}^{\infty} (1-\beta) \beta^k \bs{E} \Big[ (\bs{p}^t - \bs{p}^{t-k}) \cdot \bs{f}(\bs{p}^{t-k}) \Big] \\
		& = & \sum_{k=0}^{\infty} (1-\beta) \beta^k \bs{E} \Big[ (1-\alpha^k) (\bbs{p} - \bs{p}^{t-k}) \cdot \bs{f}(\bs{p}^{t-k}) + \sum_{s=0}^{k-1} \alpha^s \bs{\varepsilon}^{t-s} \cdot \bs{f}(\bs{p}^{t-k}) \Big] \\
		& = & \sum_{k=0}^{\infty} (1-\beta) \beta^k \bs{E} \Big[ (1-\alpha^k) (\bbs{p} - \bs{p}^{t-k}) \cdot \bs{f}(\bs{p}^{t-k}) \Big] \\
		& = & S_0 \sum_{k=0}^{\infty} (1-\beta) \beta^k (1-\alpha^k) = S_0 \Big( 1 - \dfrac{ 1- \beta }{ 1 - \alpha \beta } \Big) = \beta \left( \dfrac{ 1 - \alpha }{ 1 - \alpha \beta } \right) S_0
	\end{IEEEeqnarray*}
	and so the proof is complete. 
\end{proof}

\noindent \textbf{Proposition \ref{prop:false-forecast}.}

\begin{proof}[Proof of Proposition \ref{prop:false-forecast}.]
	It is clear that $\bs{p}^t \sim \tbs{p}^t$ as these prices both follow the same distribution conditional on $\bs{p}^{t-k}$. 
	
	Let $\bs{f}$ be the demand function which satisfies \eqref{eq:demand}. From equation \eqref{eq:pop-tmp-price-misp} of Lemma \ref{lemma:pop-tmp-price-misp}, equation \eqref{eq:S_0}, and the fact that $ ( \bs{p}^t \indep \tbs{p}^t ) | \bs{p}^{t-k}$ for all $s \in \{0,1,2,\ldots, k-1\}$
	\begin{IEEEeqnarray*}{rCl}
		\text{TMP}_0 & = & \E[ ( \bs{p}^t - \tbs{p}^t ) \cdot \bs{x}^t ] \\
		& = & \E[ \bs{p}^t \cdot \bs{x}^t ] - \E[ \tbs{p}^t \cdot \bs{x}^t ] \\
		& = & \E[ \E[ \bs{p}^t \cdot \bs{f}(\tbs{p}^t) | \bs{p}^{t-k} ] ] - \E[ \tbs{p}^t \cdot \bs{x}^t ] \\ 
		& = & \E[ \E[ \bs{p}^t | \bs{p}^{t-k} ] \cdot \E[ \bs{f}(\tbs{p}^t) | \bs{p}^{t-k} ] ] - \E[ \tbs{p}^t \cdot \bs{x}^t ] \\
		& = & \E[ ( (1-\alpha^k) \bbs{p} + \alpha^k \bs{p}^{t-k} ) \cdot \E[ \bs{f}(\tbs{p}^t)| \bs{p}^{t-k} ] ] - \E[ \tbs{p}^t \cdot \bs{x}^t ] \\
		& = & \E[ ((1-\alpha^k) \bbs{p} + \alpha^k \bs{p}^{t-k} - \tbs{p}^t ) \cdot \bs{x}^t ] \\
		& = & (1-\alpha^k) S_0 + \alpha^k \E[ ( \bs{p}^{t-k} - \tbs{p}^t ) \cdot \bs{x}^t ]
	\end{IEEEeqnarray*}
	and so \eqref{eq:false-forecast} holds. From \eqref{eq:false-forecast} it is clear that $\text{TMP}_0 \to S_0$ as $k \to \infty$. To show that $\text{TMP}_0$ is increasing in $k$ let us write $\tbs{p}_k^t$, $\bs{x}_k^t$, and $\text{TMP}_0(k)$ for mental prices, demand, and the total money pump when the consumer has not updated prices for $k$ periods. As $\tbs{p}_k^t$ and $\bs{p}^t$ are independent conditional on $\bs{p}^{t-k}$ we see from Lemma \ref{lemma:ar} that
	\begin{equation} \label{eq:mental-price1}
		\tbs{p}_k^t = (1-\alpha^k) + \alpha^k \bs{p}^{t-k} + \sum_{s=0}^{k-1} \alpha^s \tbs{\varepsilon}_k^{t-s}
	\end{equation}
	where $\tbs{\varepsilon}_k^{t-s}$ are IID mean zero shocks which are independent of $\{ \bs{\varepsilon}_j \}_{j \in \ints}$. Let 
	\begin{equation*}
		\bbs{p}_k^t = (1-\alpha^{k+1}) \bbs{p} + \alpha^{k+1} \bs{p}^{t-k-1} + \sum_{s=0}^{k-1} \alpha^s \tbs{\varepsilon}_k^{t-s}
	\end{equation*}
	From \eqref{eq:mental-price1} and the fact that $\bs{p}^{t-k}$ follows an AR(1) process we see that 
	\begin{equation} \label{eq:mental-price2}
		\tbs{p}_k^t = \bbs{p}_k^t + \bs{\varepsilon}^{t-k}
	\end{equation}
	From \eqref{eq:mental-price1} we see that
	\begin{equation} \label{eq:mental-price3}
		\tbs{p}_{k+1}^t = \bbs{p}_{k}^t + \tbs{\varepsilon}_{k+1}^{t-k}
	\end{equation}
	From \eqref{eq:mental-price2} and \eqref{eq:mental-price3} we see that $(\bs{p}^{t-k-1}, \tbs{p}_k^t) \sim ( \bs{p}^{t-k-1}, \tbs{p}_{k+1}^t )$ from which it follows that
	\begin{equation} \label{eq:demand-equality}
		\E[ \bs{p}^{t-k-1} \cdot \bs{f}( \tbs{p}_k^{t} ) ] = \E[ \bs{p}^{t-k-1} \cdot \bs{f}( \tbs{p}_{k+1}^{t} ) ]
	\end{equation}
	For a contradiction assume that $\E[ \bs{\varepsilon}^{t-k} \cdot \bs{x}_{k}^t ] > 0$. As $\bs{\varepsilon}^{t-k} \indep \bbs{p}_k^t$ we may find some $\tbs{x}^t$ which satisfies $( \tbs{x}^{t}, \bbs{p}_k^t ) \sim ( \bs{x}_k^t, \bbs{p}_k^t )$ and also $\tbs{x}^t \indep \bs{\varepsilon}^{t-k}$. Thus, using \eqref{eq:mental-price2}
	\begin{equation*}
		\E[ \tbs{p}_k^t \cdot (\bs{x}_k^t - \tbs{x}^t) ] = \E[ (\bbs{p}_k^t + \bs{\varepsilon}^{t-k}) \cdot (\bs{x}_k^t - \tbs{x}^t) ] = \E[ \bs{\varepsilon}^{t-k} \cdot ( \bs{x}_k^t - \tbs{x}^t ) ] = \E[ \bs{\varepsilon}^{t-k} \cdot \bs{x}_k^t ] > 0
	\end{equation*}
	which contradicts \eqref{eq:optimal-demand}. Having established a contradiction we see that
	\begin{equation*}
		\E[ \bs{\varepsilon}^{t-k} \cdot \bs{x}_k^t ] \leq 0
	\end{equation*}
	Using this and \eqref{eq:demand-equality} we have
	\begin{IEEEeqnarray*}{rCl}
		\text{TMP}_0(k) & = & (1-\alpha^k) S_0 + \alpha^k \E[ (\bs{p}^{t-k} - \tbs{p}_k^t) \cdot \bs{x}_k^{t} ] \\
		& = & (1-\alpha^k) S_0 + \alpha^k \E[ ( (1-\alpha) \bbs{p} + \alpha \bs{p}^{t-k-1} + \bs{\varepsilon}^{t-k} - \tbs{p}_k^t) \cdot \bs{x}_k^{t} ] \\
		& = & (1-\alpha^{k+1}) S_0 + \alpha^{k+1} \E[ (\bs{p}^{t-k-1} - \tbs{p}_k^t) \cdot \bs{x}_k^t ] + \alpha^k \E[ \bs{\varepsilon}^{t-k} \cdot \bs{x}_k^t ] \\
		& \leq & (1-\alpha^{k+1}) S_0 + \alpha^{k+1} \E[ (\bs{p}^{t-k-1} - \tbs{p}_{k+1}^t) \cdot \bs{x}_{k+1}^t ] \\
		& = & \text{TMP}_0(k+1)
	\end{IEEEeqnarray*}
	and so indeed $\text{TMP}(k)$ is increasing in $k$.
\end{proof}

\section{Online Appendix}

\subsection{Additive Cost Efficiency and Afriat's CCEI} \label{sec:cost-efficiency}

Here we show that the additive cost inefficiency $A$, as defined in Section \ref{sec:tmp}, is in fact a version of Afriat's CCEI (after applying a suitable normalization to $A$). To proceed, we normalize $A$ by dividing by total expenditure
\begin{equation*}
	\ti{A} = \inf_U \ \dfrac{\left( \sum_{t=1}^T \mbp^t \cdot \mbx^t - e_U \right)}{ \sum_{t=1}^T \mbp^t \cdot \mbx^t }
\end{equation*}
Now, let $D = (\mbp^t, \mbx^t)_{t \leq T}$ be some dataset and let $U$ be an arbitrary utility function. Let $\bar{e}_U^t$ denote the smallest amount of money which the consumer could have spent in period $t$ to acquire a bundle giving as much utility as $\mbx^t$. That is, $\bar{e}_U^t = \inf \{ \mbp^t \cdot \mbx: U(\mbx) \geq U(\mbx^t) \}$. Let $\mathcal{U}$ be some collection of utility functions. The CCEI for $\mathcal{U}$ is the number\footnote{When $\mathcal{U}$ is the class of well-behaved utility functions then $\text{CCEI}_{\mathcal{U}}$ can be expressed in terms of the extent to which the budget constraints need to be relaxed in order for the data to satisfy GARP. That is, the CCEI is equal to the infimum number $e$ so that for all $t_1,t_2,\ldots ,t_K$ with $t_1 = t_K$ we have
	\begin{equation*}
		(1-e) \mbp^{t_k} \cdot \mbx^{t_k} \geq \mbp^{t_k} \cdot \mbx^{ t_{k+1} }, \qquad\qquad \text{for all } k \in \{1,2,\ldots,K-1\}
	\end{equation*}
	implies that none of these inequalities hold strictly. See \citet{halevy2018} for the proof that our definition of the CCEI and the version involving relaxed budget sets are in fact equivalent when $\mathcal{U}$ is the collection of well-behaved utility functions.}
\begin{equation} \label{eq:ccei}
	\text{CCEI}_{\mathcal{U}} = \inf_{U \in \mathcal{U}} \sup_{t} \left( \dfrac{ \mbp^t \cdot \mbx^t - \bar{e}_U^t }{ \mbp^t \cdot \mbx^t } \right)
\end{equation}
In other words, the CCEI considers the percent of the budget set which is wasted in each period and aggregates the measure over periods by taking a supremum. 

How does the CCEI relate to $A$? To answer this question let us cease to consider $D = (\mbp^t, \mbx^t)_{t \leq T}$ as $T$ separate purchasing occasions but rather let us consider $D$ as one giant purchasing occasion in which the consumer buys the bundle $(\mbx^1, \mbx^2, \ldots, \mbx^T) \in \mathbb{R}_+^{T L}$ when prices are $(\mbp^1, \mbp^2, \ldots, \mbp^T) \in \mathbb{R}_{++}^{T L}$. Let $\mathcal{U}_{A}$ denote the collection of utility functions $V$ which take the additive form $V( \tmbx^1,\tmbx^2, \ldots, \tmbx^T ) = \sum_{t=1}^T U( \tmbx^t )$ for some $U$. It follows that  
\begin{equation*}
	e_U = \inf \left\{ \sum_{t=1}^T \mbp^t \cdot \tmbx^t: \sum_{t=1}^T U(\tmbx^t) \geq \sum_{t=1}^T U(\mbx^t) \right\} = \bar{e}_V^{1}
\end{equation*}
where $V \in \mathcal{U}_{A}$ is the additive utility function with sub-utility function $U$. Note that `1' in the superscript of $\bar{e}_V^{1}$ denotes the single (giant) observation where  $(\mbx^1, \mbx^2, \ldots, \mbx^T)$ is purchased at prices $(\mbp^1, \mbp^2, \ldots, \mbp^T)$. As $e_U = \bar{e}_V^{1}$ it is clear that $\tilde{A} = \text{CCEI}_{\mathcal{U}_A}$. Thus $\tilde A$ can be thought of as a version of Afriat's CCEI. 

\subsection{The relationship between $Q$ and $A$} \label{sec-app:Q-A}

As noted, every utility function $U$ which achieves the infimum in the definition of $Q$ also achieves the infimum in the definition of $A$. The following connects utility functions which achieve the infimum in the definition of $A$ to those which achieve the infimum in the definition of $Q$.
\begin{proposition} \label{prop:A-Q2}
	Suppose $X$ is convex and let $D = ( \bs{p}^t, \bs{x}^t )_{t \leq T}$ be a dataset. For any standard $U: X \to \ereal$ which achieves the infimum in the definition $A$ at least one of the following holds.
	\begin{enumerate}
		\item There exists a $k> 0$ so that $kU$ achieves the infimum in the definition of $Q$.
		\item The standard utility function
		\begin{equation*}
			\bar{U}(\bs{x}) = \begin{cases}
				0, \qquad\qquad & \text{ if } U(\bs{x}) \in \real \\
				-\infty, \qquad \qquad & \text{ else.}
			\end{cases}
		\end{equation*}
		achieves the infimum in the definition of $Q$.
		\item The standard utility function
		\begin{equation*}
			\ti{U}(\bs{x}) = \begin{cases}
				0, \qquad\qquad & \text{ if } U(\bs{x}) = \sup_{ \tbs{x} \in X } U(\tbs{x}) \\
				-\infty, \qquad \qquad & \text{ else.}
			\end{cases}
		\end{equation*}
		achieves the infimum in the definition of $Q$.
	\end{enumerate}
\end{proposition}
We may interpret the proposition as saying that for ``most'' $U$ which achieve the infimum in the definition of $A$ there exists a $k > 0$ so that $kU$ attains the infimum in the definition of $Q$. That is, the proposition shows that whenever $U$ achieves the infimum in $A$ then at least one of three things occurs. However, the second and third items in Proposition \ref{prop:A-Q2} are somewhat pathological. Item 2 requires that the characteristic utility function $\bar{U}$ which equals $0$ on the effective domain of $U$ attains the infimum in $Q$. Item 3 requires that the characteristic function $\ti{U}$ which is $0$ for those bundles which maximize $U$ attains the infimum of $Q$. Needless to say items 2 and 3 only occur under very special conditions.

Under certain circumstances it is impossible for either $\bar{U}$ or $\ti{U}$ to attain the infimum in $Q$. In this case we know that item 1 in Proposition \ref{prop:A-Q2} holds. For instance, if $X = \eucp$ and $U: \eucp \to \real$ is well-behaved, and consumption is not identically $\bs{0}$ (i.e.\ there exists $t$ so that $\bs{x}^t > \bs{0}$) then only item 1 can hold. To see this, note that there is no bundle $\bs{x}$ which achieves the supremum in $\sup_{\tbs{x} \in \eucp} U(\tbs{x})$ and thus $\ti{U}$ cannot achieve the infimum in the definition of $Q$. Also, $\bs{0}$ is always cheaper than every other bundle and so $\bar{U}$ cannot achieve the infimum in the definition of $Q$.

There are utility functions for which item 1 in Proposition \ref{prop:A-Q2} does not hold. For an example, suppose there is 1 good and 1 observation. In observation 1 the price is 1 and the consumer buys 0 units of consumption. The utility function $U(x) = \sqrt{x}$ additively rationalizes the behavior and so $U$ achieves the infimum in the definition of $A$. Nevertheless, the consumption $x = 0$ never maximizes the quasilinear utility function $k U(x) - x$ for any $k > 0$ (the first order condition is $k x^{-1/2} = 1$ or equivalently $x = k^2$ and so $x \neq 0$ for any $k > 0$) and so $k U$ does not attain the infimum in $Q$ for any $k > 0$. Note that $\bar{U}$ quasilinear rationalizes the behavior and thus $\bar{U}$ achieves the infimum in $Q$. 

Proposition \ref{prop:A-Q2} does not hold when $X$ is not convex. For an example, suppose there is 1 good and let $X = \{ 0,1,2 \}$. Consider the dataset $D = ( (p^1, x^1), (p^2, x^2) )$ where $p^1 = p^2 = 1$, $x^1 = 0$, and $x^2 = 2$. Let $U: X \to \real$ be defined by $U(0) = 0$, $U(1) = 1.5$, and $U(2) = 2$. Clearly this $U$ is well-behaved. The consumer spends $p^1 x^1 + p^2 x^2 = 2$ and attains additive utility $U(x^1) + U(x^2) = 2$. While the consumer purchased $(x^1,x^2) = (0,2)$ the consumer could have afforded to purchase $(0,0)$, $(1,0)$, $(0,1)$, and $(2,0)$ all of which yield weakly less additive utility than the attained utility of $2$. Therefore, $U$ additively rationalizes $D$ and thus attains the infimum in the definition of $A$. Note that neither $\bar{U}$ nor $\ti{U}$ quasilinear rationalize $U$ (and thus cannot achieve the infimum in the definition of $Q$). In particular $\bar{U}$ can only quasilinear rationalize a purchase of $0$ in both periods and $\ti{U}$ can only quasilinear rationalize a purchase of $2$ in each period. It is easy to verify that there is no $k > 0$ for which $k U$ quasilinear rationalizes $D$. In particular, when $k < 2/3$ then the consumer demands $0$, when $k = 2/3$ the consumer can demand either $0$ or $1$, when $k \in (2/3, 2)$ then the consumer demands $1$, when $k = 2$ then the consumer demand either 1 or 2, and when $k > 2$ then the consumer demands $2$. We thus see that the conclusion of Proposition \ref{prop:A-Q2} does not hold.

\begin{proof}[Proof of Proposition \ref{prop:A-Q2}.]
	Let $U$ be a standard utility function which achieves the infimum in the definition of $A$. Let $e: \real \to (-\infty, \infty]$ be defined by
	\begin{equation*}
		e( u ) = \inf \left\{ \sum_{t=1}^T \bs{p}^t \cdot \tbs{x}^t: \sum_{t=1}^T U(\tbs{x}^t) \geq u \right\}
	\end{equation*}
	As $U$ is weakly increasing and concave it is easy to see that $e$ is weakly increasing and convex. Let $A$ be the epigraph of $e$. That is, $A = \{ (u, w) \in \mathbb{R}^2: w \geq e(u) \}$. As $e$ is convex the set $A$ is convex and $(u, e(u))$ is on the boundary of $A$ for any $u$. Take $\bar{u} = \sum_{t=1}^T U(\bs{x}^t)$ and apply the supporting hyperplane theorem to see that there exists $(a,b) \neq \bs{0}$ such that 
	\begin{equation*}
		a \bar{u} - b e(\bar{u}) \geq a u - b e(u), \qquad\qquad \forall u \in \real \text{ such that } e(u) \in \real
	\end{equation*}
	As $e$ is weakly increasing we see that if $(u,w) \in A$ then $(\ti{u}, \ti{w}) \in A$ for any $(\ti{u}, \ti{w})$ satisfying $u \geq \ti{u}$ and $\ti{w} \geq w$. From which we see that $(a,b) > \bs{0}$.
	
	If $(a,b) \gg \bs{0}$ then let $k = b / a$ and note 
	\begin{equation*} 
		k \bar{u} - e(\bar{u}) \geq \sum_{t=1}^T k U(\tbs{x}^t) - \bs{p}^t \cdot \tbs{x}^t, \qquad \forall ( \tbs{x}^1, \ldots, \tbs{x}^T ) \in X^T
	\end{equation*}
	and so 
	\begin{equation*}
		\sum_{t=1}^T \sup_{ \tbs{x} \in X } \Big\{ k U(\tbs{x}) - \bs{p}^t \cdot \tbs{x} \Big\} = \sum_{t=1}^T k U(\bs{x}^t) - e(\bar{u})
	\end{equation*}
	which shows
	\begin{equation*}
		\bar{Q}(kU) = \sum_{t=1}^T \bs{p}^t \cdot \bs{x}^t - e(\bar{u}) = A
	\end{equation*}
	and so $kU$ achieves the infimum in the definition of $Q$. 
	
	Next, suppose that $a = 0$. We see that $e(u) \geq e(\bar{u})$ for all $u\in \real$. Let $\bar{U}$ be as in item 2 of the proposition and note 
	\begin{equation*}
		\sum_{t=1}^T \sup_{ \tbs{x} \in X } \Big\{ \bar{U}( \tbs{x} ) - \bs{p}^t \cdot \tbs{x} \Big\} = \sum_{t=1}^T \bar{U}(\bs{x}^t) - e(\bar{u})
	\end{equation*}
	which shows 
	\begin{equation*}
		\bar{Q}(\bar{U}) = \sum_{t=1}^T \bs{p}^t \cdot \bs{x}^t - e(\bar{u}) = A
	\end{equation*}
	
	Finally, suppose that $b = 0$. We see that $\bar{u} \geq u$ for all $u \in \real$ such that $e(u) \in \real$. Let $\ti{U}$ be as in item 3 of the proposition and note
	\begin{equation*}
		\sum_{t=1}^T \sup_{ \tbs{x} \in X } \Big\{ \ti{U}( \tbs{x} ) - \bs{p}^t \cdot \tbs{x} \Big\} = \sum_{t=1}^T \ti{U}(\bs{x}^t) - e(\bar{u})
	\end{equation*}
	which shows 
	\begin{equation*}
		\bar{Q}(\ti{U}) = \sum_{t=1}^T \bs{p}^t \cdot \bs{x}^t - e(\bar{u}) = A
	\end{equation*}
	and so the proof is complete.
\end{proof}

\end{document}